\documentclass[12pt,draftclsnofoot,onecolumn]{IEEEtran}

\usepackage{amsfonts}
\usepackage{times}
\usepackage{latexsym}
\usepackage{amssymb}
\usepackage{amsmath}
\usepackage{cite}
\usepackage{verbatim}

\newcommand{\bydef}{\triangleq}

\def\SNR{{\textsf{SNR}}}

\def\bydef{:=}

\def\bb0{{\mathbb{0}}}

\def\bydef{:=}
\def\ba{{\mathbf{a}}}
\def\bb{{\mathbf{b}}}

\def\bh{{\mathbf{h}}}

\def\bq{{\mathbf{q}}}

\def\bx{{\mathbf{x}}}
\def\by{{\mathbf{y}}}
\def\bz{{\mathbf{z}}}
\def\b0{{\mathbf{0}}}

\def\bA{{\mathbf{A}}}

\def\bD{{\mathbf{D}}}

\def\bH{{\mathbf{H}}}
\def\bI{{\mathbf{I}}}

\def\bP{{\mathbf{P}}}
\def\bQ{{\mathbf{Q}}}

\def\bS{{\mathbf{S}}}

\def\bU{{\mathbf{U}}}
\def\bV{{\mathbf{V}}}

\def\bbC{{\mathbb{C}}}

\def\bbE{{\mathbb{E}}}

\def\bbN{{\mathbb{N}}}

\def\bbR{{\mathbb{R}}}

\def\bydef{:=}

\def\sf0{{\mathsf{0}}}

\usepackage{graphicx}
\usepackage{amssymb}
\usepackage{amsfonts}
\usepackage{amsmath}
\usepackage{latexsym}

\begin{document}

\newtheorem{thm}{Theorem}
\newtheorem{lemma}{Lemma}
\newtheorem{rem}{Remark}
\newtheorem{exm}{Example}
\newtheorem{prop}{Proposition}
\newtheorem{defn}{Definition}
\newtheorem{cor}{Corollary}
\def\proof{\noindent\hspace{0em}{\itshape Proof: }}
\def\endproof{\hspace*{\fill}~\QED\par\endtrivlist\unskip}
\def\bh{{\mathbf{h}}}
\def\SNR{{\mathsf{SNR}}}
\def\SIR{{\mathsf{SIR}}}
\title{Transmission Capacity of Ad-hoc Networks with Multiple Antennas using Transmit Stream Adaptation and Interference Cancelation}
\author{
Rahul~Vaze and Robert W. Heath Jr.
\thanks{Rahul~Vaze is with the School of Technology and Computer Science, Tata Institute of Fundamental Research, Homi Bhabha Road, Mumbai 400005, vaze@tcs.tifr.res.in.  \newline Robert W. Heath Jr. is with the Wireless Networking and Communications Group, Department of Electrical and Computer Engineering, The University of Texas at Austin, 1 University Station, C0803, Austin, TX 78712-0240, rheath@ece.utexas.edu. \newline This work was funded by DARPA through IT-MANET grant no. W911NF-07-1-0028.}}

\date{}
\maketitle
\noindent
\begin{abstract}
The transmission capacity of an ad-hoc network is the maximum density of active
 transmitters per unit area, given an outage constraint at each receiver for
a fixed rate of transmission. Assuming that the transmitter locations are distributed as a Poisson point process, this paper derives upper and lower bounds on the transmission
capacity of an ad-hoc network when each node is equipped with multiple
antennas. The transmitter either uses eigen multi-mode
beamforming or a subset of its antennas to transmit multiple data streams,
while the receiver uses partial zero forcing to cancel certain
interferers using some of its spatial receive degrees of freedom (SRDOF). The receiver either cancels the nearest interferers or those interferers that maximize the post-cancelation signal-to-interference ratio.
Using the obtained bounds, the optimal number of data streams to transmit,
and the optimal SRDOF to use for interference cancelation are derived that provide the best scaling of the transmission capacity with the number of antennas.
With beamforming, single data stream transmission together with using all but one SRDOF for interference cancelation  is  optimal, while without beamforming, single data stream transmission
together with using a fraction of the total SRDOF for  interference cancelation is  optimal. \end{abstract}

\section{Introduction}
In an ad-hoc wireless network, multiple transmitter-receiver pairs
communicate simultaneously  without the help of
any fixed infrastructure. 
Inter-user interference is a major bottleneck in an ad-hoc wireless network,
severely limiting the rate of successful transmissions.
One way to quantify the performance in an ad-hoc network is through the notion of the transmission capacity. Defined in \cite{Weber2005}, and subsequently studied in \cite{Weber2007, Weber2008, Baccelli2006, Haenggi2009}, the transmission capacity of an ad hoc network is 
 the maximum allowable transmission density of nodes, satisfying a per transmitter receiver rate, and outage probability constraints.
The transmission capacity characterizes
the maximum density of spatial transmissions that can be
 supported  simultaneously in an ad hoc network under a quality
of service constraint.

Employing multiple antennas at each node is one way to manage interference and increase data rate in an ad-hoc wireless network.  For example, multiple antennas can be used to increase  the per-link rate through spatial multiplexing, or to increase spatial diversity for  reducing fading outages, or for receiver interference cancelation to remove strong interferers. The transmission capacity of some specific multiple antenna strategies like beamforming, maximum ratio combining (MRC), spatial multiplexing, and zero forcing  have been derived in prior work \cite{Hunter2008, Huang2008, Jindal2008a, Mackay2009, Proakis2007,Ali2010}. The general problem of finding the optimal use of multiple transmit 
and receive antennas to maximize the transmission capacity has however remained unsolved.

To characterize the optimal use of multiple antennas in ad hoc networks, in this paper we derive upper and lower bounds on the transmission capacity 
in multiple antenna ad hoc networks, with multi-stream transmission and interference cancelation at the receiver. 
We assume that the transmitter locations are distributed as a Poisson point process (PPP), and each node of the ad-hoc network is equipped with $N$ antennas for transmission and reception. 

We consider two transmission strategies: multi-mode spatial multiplexing without channel state information at the transmitter (CSIT) \cite{Foschini1998}, and multi-mode beamforming with CSIT \cite{Telatar1999}. 
We assume that each receiver uses partial zero forcing (ZF), where some of the spatial receive degrees of freedom (SRDOF) are used for decoding the signal of interest leaving the remaining SRDOF for interference cancelation. 
We derive results when each receiver cancels the nearest interferers in terms of their distance from the receiver, or cancels the interferers that maximize the post-cancelation signal-to-interference ratio (SIR).
 Our results are summarized as follows. 
\begin{itemize}
\item Spatial Multiplexing  (without CSIT)
\begin{itemize}
\item {\bf Canceling the nearest interferers in terms of their distance from the receiver, or the interferers that maximize the post-cancelation SIR}: Transmitting a single data stream together with using a fraction of total  SRDOF
for interference cancelation provides the best scaling of the
transmission capacity with respect to $N$; the transmission capacity lower bound scales linearly with $N$. \footnote{The results for canceling the nearest interferers without transmit beamforming have appeared in \cite{VazeAsilomar2009} in part.}
\end{itemize}
\item With Transmit Beamforming (with CSIT)
\begin{itemize}
\item
{\bf Canceling the nearest interferers in terms of their distance from the receiver}: Single stream beamforming together with using $N-1$ SRDOF for
interference cancelation  provides the best scaling of the
transmission capacity with respect to $N$; the
transmission capacity lower bound scales linearly with $N$. \footnote{The results for canceling the nearest interferers  with transmit beamforming have appeared in \cite{VazeSPAWC2009} in part.}

\end{itemize}

\end{itemize}

The differences between our paper and prior work are summarized as follows.  Without CSIT, the transmission capacity has been analyzed for single transmit antenna with no interference cancelation \cite{Hunter2008}, multiple transmit antennas with no interference cancelation \cite{Mackay2009,Proakis2007},  
single transmit antenna with canceling $N-1$ interferers \cite{Huang2008}, and single transmit antenna and using a fraction of total  SRDOF for interference cancelation  \cite{Jindal2008a}. 
In this paper we consider multi-stream transmission, unlike \cite{Hunter2008,Huang2008,Jindal2008a},
and canceling a fraction of the received interferers, generalizing  \cite{Hunter2008,Huang2008,Jindal2008a,Mackay2009,Proakis2007}. 
Without CSIT, and when receiver employs interference cancelation, we show that it is optimal to transmit a
single data stream and use a fraction of the total  SRDOF for interference cancelation. Our work shows that the strategy proposed in \cite{Jindal2008a} is transmission-capacity scaling optimal in terms of the number of antennas.


With CSIT, the transmission capacity has been computed for single stream beamforming without interference cancelation in \cite{Hunter2008}.  We generalize  \cite{Hunter2008} by considering interference cancelation together with multi-mode beamforming, where multiple data
streams are sent by the transmitter on multiple eigenmodes of the channel. Our results show that using interference cancelation at the receiver 
in conjunction with single stream beamforming, 
the transmission capacity scales linearly with $N$ in contrast to sublinear
scaling without interference cancelation \cite{Hunter2008}.

{\it Notation:} Let ${\bA}$ denote a matrix, ${\bf a}$ a vector and
$a(i)$ the $i^{th}$ element of ${\bf a}$.   The field of real
and complex numbers is denoted by $\bbR$ and $\bbC$, respectively.
The space of $M\times N$ matrices with complex entries is denoted by
${\bbC}^{M\times N}$. An $N \times N$ identity matrix is represented by $\bI_N$. The Euclidean norm of a vector $\bf a$ is
denoted by $|\ba|$.  The superscripts $^T, ^*$ represent the
transpose, and the  transpose conjugate, respectively. The
expectation of a function $f(x)$ of random variable  $x$ is denoted by
${\bbE}\{f(x)\}$. The integral $\int_{0}^{\infty}x^{k-1}e^{-x}dx$ is
denoted by $\Gamma(x)$.
A circularly symmetric complex Gaussian random
variable $x$ with zero mean and variance $\sigma^2$ is denoted as $x
\sim {\cal CN}(0,\sigma^2)$. The factorial of an integer $n$ is denoted
as $n!$. Let $S_1$ be a set, and $S_2$ be a subset of $S_1$. Then
$S_1 \backslash S_2$ denotes the set of elements of $S_1$ that do not belong to $S_2$. The cardinality of any set $S$ is denoted by $|S|$. Let $f(n)$ and $g(n)$ be two function defined on some subset of real numbers.
Then we write $f(n) = \Omega(g(n))$ if
$\exists \ k > 0, \ n_0, \ \forall \ n>n_0$, $|g(n)| k \le |f(n)|$,
$f(n) = {\cal O}(g(n))$ if $\exists \ k > 0, \ n_0, \ \forall \ n>n_0$, $|f(n)| \le |g(n)| k$, and
$f(n) = \Theta(g(n))$ if $\exists \ k_1, \ k_2 > 0, \ n_0, \ \forall \ n>n_0$,
$|g(n)| k_1 \le |f(n)| \le |g(n)| k_2$. We use the symbol
$\bydef$ to define a variable.

{\it Organization:} The rest of the paper is organized as follows.
In Section \ref{sec:sys}, we describe the system model under consideration. 
In Section \ref{sec:csir}, upper and lower bounds on the transmission capacity are derived for the case when the transmitter sends multiple independent data using spatial multiplexing  and the receiver uses partial ZF decoder. In Section \ref{sec:csit}, upper and lower bounds on the transmission capacity  are derived for the case when the transmitter uses
multi-mode beamforming, and the receiver uses partial ZF decoder. Numerical results are illustrated in 
Section \ref{sec:sim} followed by conclusions  in Section \ref{sec:conc}.

\section{System Model}
\label{sec:sys} Consider an ad-hoc network where
each node is equipped with $N$ antennas for transmission and reception. We adopt the assumptions considered in previous transmission capacity analysis of ad-hoc
networks \cite{Weber2005, Weber2007, Huang2008}. The location of each
source is modeled as a homogenous PPP on a two-dimensional plane with intensity $\lambda_0$. Thus, the mean number of sources in an unit area is $\lambda_0$.
Each source node  communicates with one destination  located at a fixed distance $d$ away.
We consider a slotted ALOHA like random access protocol, where each source  attempts to transmit with an access probability $p_a$, independently of all other transmitters. \footnote{Other more intelligent MAC strategies such as not allowing interferers inside a guard zone to transmit have been considered in prior work for computing the transmission capacity \cite{Hasan2007}. Their analysis, however,  is quite complicated and is outside of the scope of this paper.} 
An active source is referred to as a transmitter, and a destination associated with a transmitter  is referred to as a receiver. Consequently, the transmitter process is also a homogenous
PPP on a two-dimensional plane with intensity
$\lambda = p_a\lambda_0$. Let the location of the $n^{th}$ 
transmitter be $T_n$ for $n\in \bbN$. The set of all 
transmitters is denoted by $\Phi = \left\{T_n\right\}$. 
Following  \cite{Weber2005}, we consider a typical
transmitter receiver pair $(T_0, R_0)$ to compute the transmission capacity,  since from the stationarity of the homogenous PPP, and  Slivnyak's Theorem (Page $121$) \cite{Stoyan1995}, it follows that the statistics of the signal received at the typical receiver are identical to that of any other receiver.



%
\subsection{Signal Model}
Let $\bx_n = [x_{n}(1), \dots, x_{n}(k)]^T, \ k \in \left[1,2,\ldots, N\right]$, denote the  data stream vector
to be sent by transmitter $T_n$ to its receiver $R_n$, where each $x_{n}(m), \ m=1,\dots,k$ is i.i.d. ${\cal CN}(0,P/k)$ distributed, and $P$ is the average power constraint at the each transmitter. 
Let  $\bH_{0n}\in \bbC^{N\times k}$
be the channel coefficient matrix between $T_n$ and $R_0$ with  i.i.d. ${\cal CN}(0,1)$ distributed entries, and 
$d_n$ be the distance between $T_n$ and $R_0$. 
We assume that  the path loss exponent $\alpha > 2$.

We consider two transmission strategies. Under a no CSIT assumption, we consider spatial multiplexing where transmitter $T_n$ sends 
vector $\bP_n \bx_n$, where $\bP_n=\bI_N^k$, where  $\bA^{k}$ denotes  the matrix consisting of first $k$ columns of
any matrix $\bA$. Under a CSIT assumption, we suppose that the transmitter $T_n$ knows the channel $\bH_{nn}$ between itself and its corresponding receiver $R_n$, and sends $\bP_{n}\bx_n$, where $\bP_{nn} = \bV_{nn}^k$, and  
the singular value decomposition of $\bH_{nn} \bydef \bU_{nn}\bD_{nn}\bV_{nn}^{*}$. 
We use the path loss model of $D^{-\alpha}$, if the Euclidean distance between any two nodes is $D$.


We describe the signal model in detail for the typical transmitter-receiver pair $(T_0, R_0)$. The received signal at the typical receiver $R_0$ is
\begin{equation}
\label{eqcsirrxmodel}
\by_0= d^{-\frac{\alpha}{2}}\bH_{00}\bP_0\bx_0 + \sum_{n: T_n \in \Phi \backslash \{T_0\}}^{\infty}d_n^{-\frac{\alpha}{2}}\bH_{0n}\bP_n\bx_n + \bz_0,
\end{equation}
where $\bz_0$ is the additive white Gaussian noise. We consider the interference limited regime, i.e.
the noise power is negligible compared to the interference power, and
henceforth ignore the noise contribution.
We assume that each entry of $\bH_{0n}$, and $\bH_{nn}$ is
independent and identically distributed (i.i.d.) ${\cal CN}(0,1)$ to model a richly
scattered fading channel with independent fading coefficients between different
transmitting receiving antennas similar to \cite{Hunter2008,Huang2008,Jindal2008a}.

\begin{figure}
\centering
\includegraphics[width=4in]{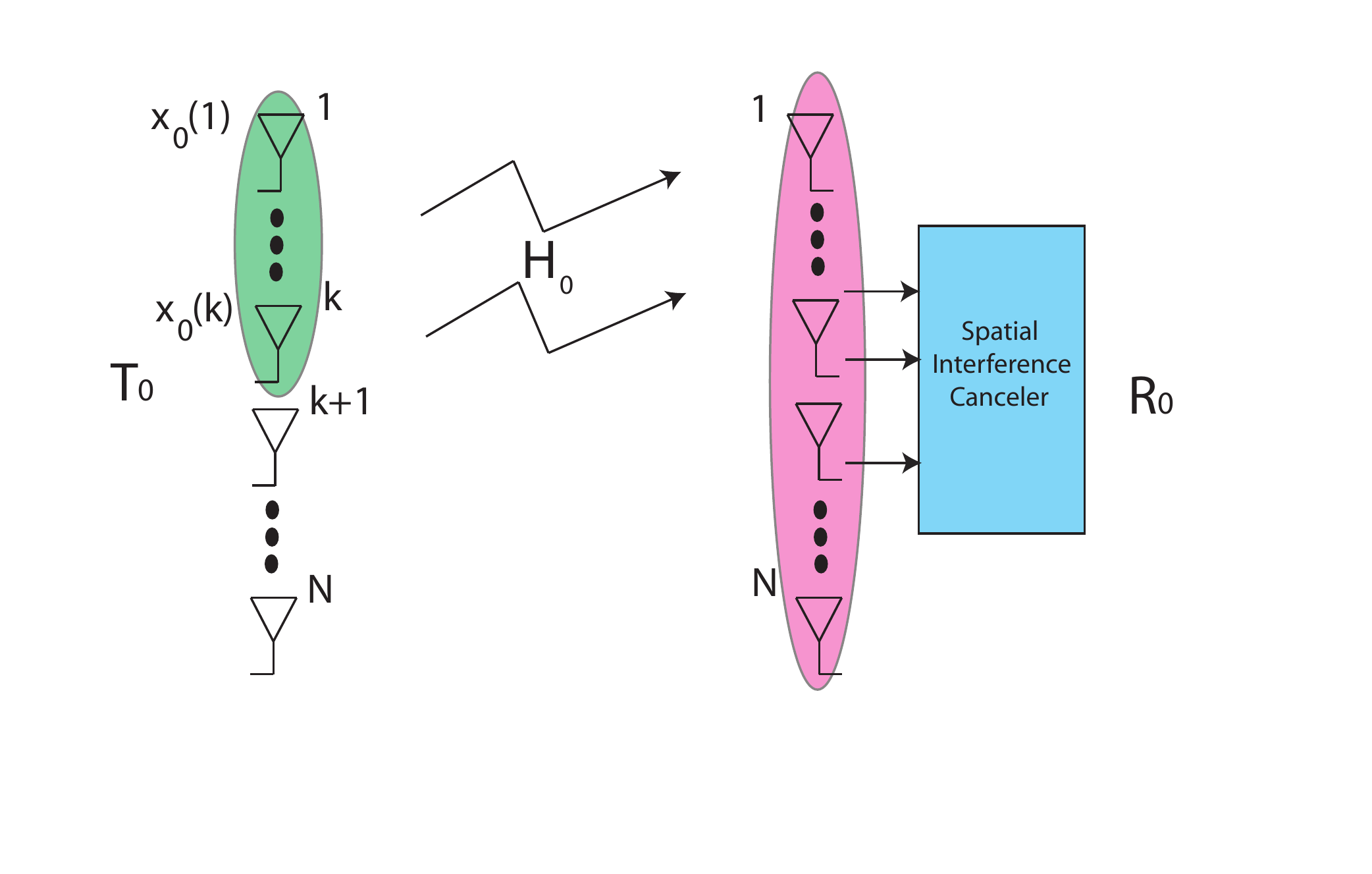}
\caption{System block diagram with no CSI at the transmitter and spatial interference cancelation at the receiver}
\label{fig:blknocsit}
\end{figure}

The decoding strategy used at each receiver depends on the transmit signaling assumption.  
Without CSIT,  where $\bP_n = \bI_N^k$, we assume that the receiver performs single stream decoding using partial ZF \cite{Jindal2008a}, where the receiver $R_0$ uses 
$N-m$ SRDOF for decoding any data stream $x_{0}(j), \ j=1,2,\dots,k$, leaving the remaining $m$ SRDOF for
canceling  the $c(k,m) \bydef \left\lfloor\frac{m}{k}\right\rfloor$ interferers.\footnote{If $c(k,m) < \frac{m}{k}$ due to the floor, then the receiver uses
use $k \times c(k,m) < m$ SRDOF for cancelation, leaving $N - k c(k,m) > N-m$ for decoding the signal of interest.} 
Partial ZF allows to keep the analysis tractable while incurring low decoding complexity. 
With partial ZF, to decode data stream $x_0(j)$ 
at $R_n$, all the other data streams $x_{0}(1), \dots \ x_{0}(j-1), \ x_{0}(j+1), \dots \ x_{0}(k)$ sent from transmitter $T_0$ also appear as interference. Therefore effectively only $N-m-k+1$ SRDOF are used to decode any data stream $x_0(j)$.  The transmit receive strategy without CSIT is depicted in Fig. \ref{fig:blknocsit}. 

For analytical purposes, we assume that the number of canceled interferers $c(k,m) > \frac{\alpha}{2}-1$. Since the typical range of the path-loss exponent $\alpha$ is between $2$ and $4$, $c(k,m) > \frac{\alpha}{2}-1$ implies that at least one interferer should be canceled. Thus our analysis precludes the case of no interference cancelation, 
which has already  been studied in \cite{McKayTC2011, Hunter2008}. Moreover, since we are interested in finding the optimal scaling of the transmission capacity with the number of antennas $N$ at each node, the constraint 
$c(k,m) > \frac{\alpha}{2}-1$ is not that restrictive, since for large values of $N$, there is sufficient flexibility for choosing optimal $k, m$, with $c(k,m) > \frac{\alpha}{2}-1$.

Let ${\cal S}_0 \subset \Phi\backslash \{T_n\}$ be the  subset of interferers to be canceled at $R_0$ with $|{\cal S}_0| = c(k,m)$. 
With partial ZF, matrix $\bQ_0 = [\bq_{1}(1) \bq_{0}(2) \dots \bq_{0}(k) ]^T$ is multiplied to the received signal $\by_0$, where 
$\bq_{0}(\ell) \in \bbC^{N \times 1}$ lies in the null space ${\cal N} ( {\cal H}_{0\ell} )$ of the matrix 
\[{\cal H}_{0\ell} \bydef \left[\bH_{00}(1) \ldots \bH_{00}(\ell-1)\ \bH_{00}(\ell+1) \ldots \bH_{00}(k) \ \bH_{0{\cal S}_0(1)} \ \bH_{0{\cal S}_0(2)} \ldots \bH_{0{\cal S}_0(c(k,m))}\right],\] where $\bH_{00}(p)$ represents the $p^{th}$ column of 
$\bH_{00}$, and ${\cal S}_0(j)$ is the $j^{th}$ element of ${\cal S}_0$.
Multiplying $\bQ_0$ to $\by_0$, we get ${\hat \by_0} \bydef \bQ_0 \by_0$, where the 
$\ell^{th}$ element of ${\hat \by_0}$ is 
\begin{eqnarray}
{\hat \by_0}(\ell) =  d^{-\frac{\alpha}{2}}\bq_{0\ell}^T \bH_{00} x_0(\ell) + 
\sum_{n:T_n \in \Phi\backslash\{{\cal S}_0, \{T_0\}\}}d_n^{-\frac{\alpha}{2}}\sum_{j=1}^k\bq_{\ell}^T\bH_{0n}(j)x_{n}(j).
\end{eqnarray}
Therefore without CSIT, the SIR for the $\ell^{th}$ data stream from $T_0$ while canceling the interferers belonging to ${\cal S}_0$ is 
\begin{equation}\label{eq:sircsir}\SIR_{{\cal S}_0, \ell} \bydef  \frac{d^{-\alpha} |\bq_{0\ell}^T\bH_{00}(\ell)|^2}{\sum_{n:T_n\in \Phi \backslash \{\{T_0\}, \cal S\}} d_n^{-\alpha}\sum_{j=1}^k|\bq^{T}_{0\ell}\bH_{0n}(j)|^2}.
\end{equation}


\begin{figure}
\centering
\includegraphics[width=4in]{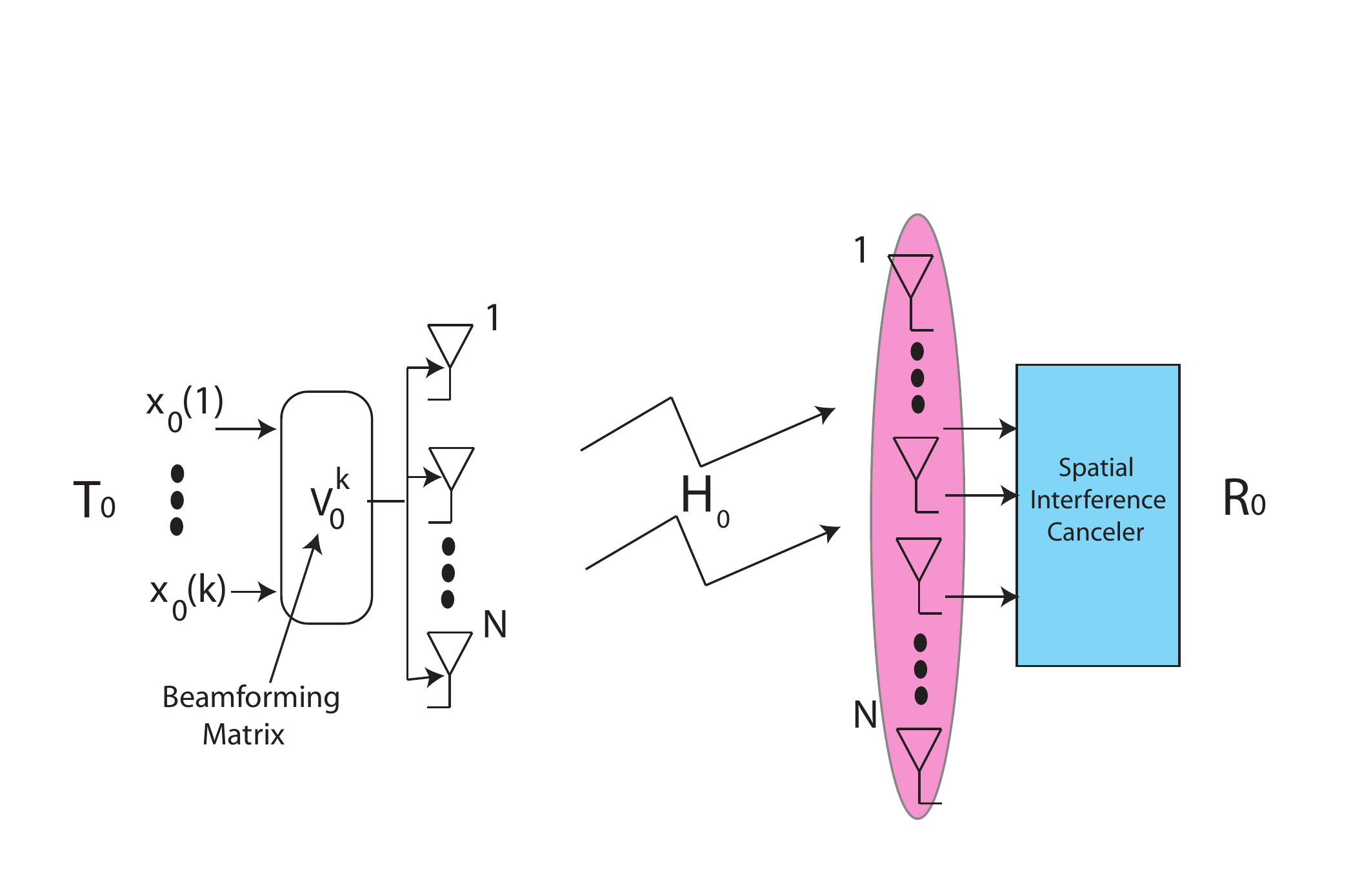}
\caption{System block diagram with transmit beamforming and interference cancelation at the receiver}
\label{fig:blkcsit}
\end{figure}

With CSIT, recall that we consider multi-mode beamforming for transmission $\bP_n = \bV_{nn}^k$, where $\bV_{nn}$ is the matrix of the right singular vectors of matrix $\bH_{nn}$. 
With CSIT, we assume that each receiver uses $k$ SRDOF to receive the intended signal, while the remainder 
of $N-k$ SRDOF are used to cancel the $c(k) \bydef \left\lfloor\frac{N}{k}\right\rfloor -1$  interferers. Similar to the case without CSIT, for the purposes of analysis we assume that $c(k) > \frac{\alpha}{2}-1$.
A block diagram depicting the
transmit-receive strategy with CSIT is illustrated in Fig. \ref{fig:blkcsit}.
Let ${\cal S}_0 \subset \Phi \backslash \{T_n\}$ be the subset of interferers to be canceled at $R_0$ with $|{\cal S}_0| = c(k)$. 
Let $\bS_0$ be the basis of the null space of the $c(k)$  interferers of ${\cal S}_0$ to be
canceled at $R_0$. Since $N-k$
SRDOF are used for interference cancelation, $\bS_n \in \bbC^{k\times N}$.
For the typical receiver $R_0$, multiplying $\bS_0$ to the received signal,
\begin{eqnarray*}
\bS_0\by_0 &=&d^{-\frac{\alpha}{2}}\bS_0 \bU_{00}\bD_{00}\bV_{00}^{*} \bV_{00}^k \ \bx_0 + \sum_{n:\ T_n \in \Phi \backslash \{\{T_0\}, {\cal S}_0\}} d_n^{-\frac{\alpha}{2}}\bS_0 \bH_{0n} \bV_{nn}^k \bx_n, \\
&= &d^{-\frac{\alpha}{2}}\bS_0 \bU^{k}_{00}\bD_{00}^{k}\bx_0 + \sum_{n: \ T_n \in \Phi \backslash \{\{T_0\}, {\cal S}_0\}} 
d_n^{-\frac{\alpha}{2}}\bS_0 \bH_{0n} \bV_{nn}^k \bx_n,
\end{eqnarray*}
where $\bU^{k}_{00}$ is the $N\times k$ matrix consisting of the first $k$ columns of $\bU_{00}$, and
$\bD_{00}^{k}\in \bbC^{k\times k}$ is the diagonal matrix consisting of
the first $k$ rows and $k$ columns of $\bD_{00}$.
Since $\bS$ and $\bU^{k}_{00}$ are both of rank $k$, and are
independent  with each entry drawn from a continuous
distribution, $\bS\bU^{k}_0$ is full rank with probability $1$.
Multiplying $\left(\bS\bU^{k}_{00}\right)^{-1}$ to the
received signal,
\begin{eqnarray*}
{\hat \by}_0&= &d^{-\frac{\alpha}{2}}\bD_{00}^{k}\bx_0 + \sum_{n:\ T_n \in \Phi \backslash \{T_0, {\cal S}\}} 
d_n^{-\frac{\alpha}{2}}\left(\bS\bU^{k}_0\right)^{-1}\bS \bH_{0n} \bV_{nn}^k \bx_n.
\end{eqnarray*}

Note that $\bD_{00}$ is the diagonal matrix of the singular values of
$\bH_{00}$.
Denoting the $\ell^{th}$ eigenvalue of $\bH_{00}\bH_{00}^{*}$ by $\gamma_{\ell}$, with $\gamma_1 \ge \gamma_2 \ge \dots$, the received signal can be separated in terms of $\bx_0(\ell), \ \ell =1,2,\ldots,k$ as
\begin{eqnarray}\label{rxsigortho}
{\hat y}_0(\ell)&= &d^{-\frac{\alpha}{2}}\sqrt{\gamma_{\ell}}x_0(\ell) + \sum_{n:\ T_n \in \Phi \backslash \{T_0, {\cal S}\}}
d_n^{-\frac{\alpha}{2}}\sum_{j=1}^k c^{n}_{\ell j} x_n(j), \ \ell =1,2,\ldots,k,
\end{eqnarray}
 where  $c^{n}_{\ell j}$ is the $(\ell, j)^{th}$ element of
$\left(\bS\bU^{k}_0\right)^{-1}\bS \bH_{0n} \bV_{nn}^k$. Note that ${\hat y}_0(\ell)$ has no contribution from $x_0(1), \ldots, x_{0}(\ell-1), x_{0}(\ell+1), \ldots, x_{0}(k),\  \ell =1,2,\ldots,k$.  Thus, multi-mode beamforming removes the intra-stream interference, and provides $k$ parallel channels between each transmitter and receiver.
Let $\mu_n^{\ell} \bydef \sum_{j=1}^k |c^{n}_{\ell j}|^2$.
Therefore with CSIT, using multi-mode beamforming,  the SIR for the $\ell^{th}$ data stream from $T_0$ while canceling interferers belonging to ${\cal S}_0$ is 
\begin{equation}\label{eq:sircsit}
\SIR^{BF}_{{\cal S}_0, \ell} \bydef  \frac{d^{-\alpha} \gamma_{\ell}}{\sum_{n:T_n\in \Phi \backslash \{\{T_0\}, {\cal S}_0\}} d_n^{-\alpha} \mu_n }.
\end{equation}


\subsection{Interference Cancelation Algorithms}
We consider two different choices of selecting ${\cal S}_0$, the subset of interferers to cancel, that are described as follows.
\begin{enumerate}
\item  Canceling the $c(k,m)$ (without CSIT) or $c(k)$ (with CSIT) nearest  interferers in terms of distance from the receiver. Let the indices of the interferers be
sorted in an increasing order in terms of their distance $d_n$ from $R_0$, i.e.
$d_1 \le d_2 \le \ldots \le d_{c(k,m)} \le d_{c(k,m)+1} \le \ldots$. Then the subset of interferers to cancel is ${\cal S}_0 = \{1,2,\ldots, c(k,m)\}$ (without CSIT) or ${\cal S}_0 = \{1,2,\ldots, c(k)\}$ (with CSIT).

Without CSIT, while canceling the $c(k,m)$ nearest interferers, the optimal $\bq_{\ell}^T$ that maximizes the signal power  $s\bydef |\bq_{\ell}^T\bH_{00}(\ell)|^2$ is $ \bq_{\ell}^T = \frac{\bH_{00}(\ell)^*\bS\bS^*}{|\bH_{00}(\ell)^*\bS\bS^*|}$ \cite{Jindal2008a}, where  $\bS$ is the basis of the null space of the canceled interferers. Moreover, as shown in \cite{Jindal2008a}, $s$ is Chi-square distributed with $2(N-k-m+1)$ DOF when each channel coefficient is Rayleigh distributed. Note  that since $\bq^T_{\ell}$ is independent of $ \bH_{0n}$ for $n \notin {\cal S}_0$, 
 $\rho_n^{\ell} \bydef \sum_{j=1}^k|\bq^T_{\ell}\bH_{0n}(j)|^2$
is Chi-square distributed with $2k$ DOF.
Thus, without CSIT, from (\ref{eq:sircsir}), the SIR for the
$\ell^{th}$ stream is given by
\begin{equation}\label{eq:sircsirnc}
\SIR_{{\cal S}_0, \ell} = \frac{d^{-\alpha} s }
{\sum_{n:T_n\in \Phi \backslash \{\{T_0\}, {\cal S}_0\}} d_n^{-\alpha} \rho_n}.
\end{equation}

While with CSIT, from (\ref{eq:sircsit}), the SIR for the $\ell^{th}$ data stream  is 
 \begin{equation}\label{eq:sircsitnc}
 \SIR^{BF}_{{\cal S}_0, \ell} \bydef  \frac{d^{-\alpha} \gamma_{\ell}}{\sum_{n:T_n\in \Phi \backslash \{\{T_0\}, {\cal S}_0\}} d_n^{-\alpha} \mu_n },
 \end{equation}
 where $\mu_n^{\ell} = \sum_{j=1}^k |c^{n}_{\ell j}|^2 $, and  ${\cal S}_0$ is the set of $c(k)$  nearest interferers.
Since $\bS_0$, $\bU_{00}^k$ and $\bV_{nn}$ are
independent of $\bH_{0n}, \ n \notin \{{\cal S}_0, \{0\}\}$, and each channel coefficient is Rayleigh distributed, using the definition of $\mu_n^{\ell}$, it follows from \cite{Jindal2008a} that $\mu_n^{\ell}$ is a Chi-square distributed random variable with $2k$ DOF $\forall \ n, \ \ell$. 

\item  Without CSIT,  from (\ref{eq:sircsir}), with partial ZF and canceling interferers from the set ${\cal S}_0$, the SIR for  data stream $\ell$ at $R_0$ is 
\[ \SIR_{{\cal S}_0, \ell} \bydef \frac{d_n^{-\alpha} |\bq_{\ell}^T\bH_{00}(\ell)|^2}{\sum_{n:T_n\in \Phi \backslash \{\{T_0\}, {\cal S}_0\}} d_n^{-\alpha}\sum_{j=1}^k|\bq^{T}_{\ell}\bH_{0n}(j)|^2}.\] 
To maximize $\SIR_{{\cal S}_0, \ell}$,  the optimal set of interferers to cancel ${\cal S}_0$, and the optimal $\bq_{\ell}$ is given by the solution of  the following optimization problem:
\begin{equation}\label{eq:sirmax} \arg \max_{\bq_{\ell}, \ {\cal S}_0, \ {\cal S}_0 \subset \Phi, \  |{\cal S}_0| = c(k,m)} \SIR_{{\cal S}_0, \ell}.
\end{equation}

The transmission capacity analysis with the optimal interference cancelation algorithm is hard, and we are not able to find closed form 
analytical transmission capacity with the optimal ordering. See Remark \ref{restrictionSIR} for more discussion. 
For analytical tractability,  we add an additional constraint that if ${\cal S}_0$ is the set of interferers chosen for cancelation, and $\bS_0$ is the orthonormal basis of the null space of ${\cal S}_0$, then 
$\bq_{\ell}^T =  \frac{\bH_{00}(\ell)^*\bS_0\bS_0^*}{|\bH_{00}(\ell)^*\bS_0\bS_0^*|}$. This choice of $\bq_{\ell}$ is 
motivated by the fact that it maximizes  the signal power $|\bq_{\ell}^T\bH_{00}(\ell)|^2$ when $\bq_{\ell} \in {\cal N}({\cal S}_0)$ \cite{Jindal2008a}.  
With this extra constraint, $S_{\ell}^{\star}$ is the optimal set of interferers to cancel where 
\begin{equation}\label{eq:msir}
{\cal S}_{0 \ell}^{\star} = \arg \max_{{\cal S}_0, \ {\cal S}_0 \subset \Phi, \  |{\cal S}_0| = c(k,m), \  \bq_{\ell}^T =  \frac{\bH_{00}(\ell)^*\bS_0\bS_0^*}{|\bH_{00}(\ell)^*\bS_0\bS_0^*|}} \SIR_{{\cal S}_0, \ell} . 
\end{equation} 
We call this the constrained maximum SIR (CMSIR) algorithm. 
Note that similar to the case of canceling the nearest interferers, $\bq_{\ell}^T =  \frac{\bH_{00}(\ell)^*\bS_0\bS_0^*}{|\bH_{00}(\ell)^*\bS_0\bS_0^*|}$ with the CMSIR algorithm as well, and hence the signal power $s= |\bq_{\ell}^T\bH_{00}(\ell)|^2$ with the CMSIR algorithm is also Chi-square distributed with $2(N-k-m+1)$ DOF. Moreover, since 
 $\bq^T_{\ell}$ is independent of $ \bH_{0n}$ for $n \notin {\cal S}^{\star}_{0 \ell}$, $\rho_n^{\ell} = |\bq^T_{\ell}\bH_{0n}(j)|^2$ is Chi-square distributed with $2$ DOF.

 Thus, without CSIT, from (\ref{eq:sircsir}), the SIR for the
$\ell^{th}$ stream using the CMSIR algorithm is given by
\begin{equation}\label{eq:sirmsir}
\SIR^{CMSIR}_{{\cal S}_0, \ell} = \frac{d^{-\alpha} s }
{\sum_{n:T_n\in \Phi \backslash \{\{T_0\}, {\cal S}^{\star}_{0 \ell}\}} d_n^{-\alpha} \rho_n}.
\end{equation}

  \end{enumerate}
  
\begin{rem}\label{restrictionSIR}
The CMSIR algorithm (\ref{eq:msir}) is restrictive since we have fixed $\bq_{\ell}^T =  \frac{\bH_{00}(\ell)^*\bS\bS^*}{|\bH_{00}(\ell)^*\bS\bS^*|}$, where $\bS$ is the orthonormal basis of the null space of interferers to be canceled. 
The advantage of this restriction is that the canceling vector $\bq_{\ell}^T$ does not depend on $\bH_{0n}, n\in \{\Phi \backslash \{ \{0\}, {\cal S}_{0 \ell}^\star\}\}$, and consequently the signal and interference powers after interference cancelation are not correlated. 
In general, with unrestricted SIR maximization (\ref{eq:sirmax}), $\bq_{\ell}^T$ could 
depend on $\bH_{00}$, $\bH_{0n}, n\in {\cal S}_{0\ell}^\star$ as well as $\bH_{0n}, n\in \{\Phi \backslash {\cal S}_{\ell}^\star\}$, and consequently the signal and interference powers are correlated, and finding the distribution of the signal and the interference power after interference cancelation is challenging.
Even though the CMSIR algorithm is restrictive, by definition it is better than canceling the nearest interferers, 
  since canceling the nearest interferers lies in the feasible set of the optimization problem (\ref{eq:msir}) solved with the CMSIR algorithm.
  \end{rem}
    \begin{rem} (CSI requirement:) For canceling the nearest interferers, each receiver requires  information about the distance of each interferer, and  CSI for only the nearest interferers it wishes to cancel. Since the distances of the interferers are assumed to vary slowly with time, it is easy to acquire the distance information with low overhead. In practical systems this can be done by averaging the received signal strength indicator for example. In contrast, the CMSIR algorithm, requires CSI from all the interferers for finding the  SIR maximizing subset ${\cal S}^{\star}$. Thus, the CMSIR algorithm is computationally expensive as compared to canceling the nearest interferers, however, it provides better  performance compared to canceling the nearest interferers.
  \end{rem}

\subsection{Problem Formulation}

The original definition of the transmission capacity with a single transmit antenna is 
$C^{siso}_{\epsilon} \bydef \lambda(1-\epsilon)R$ \cite{Weber2005}, where $\lambda$ is the maximum density of transmitters per unit area such that the outage probability at any receiver is less than  $\epsilon$ with rate of transmission $R$ bits/sec/Hz. 
With an outage probability constraint of $\epsilon$ at rate $R$ bits/sec/Hz, the average throughput of each transmitter-receiver link is $(1-\epsilon)R$, and the 
transmission capacity metric accounts for the average network throughput by adding the average throughput of all the $\lambda$ transmitters per unit area. 
We take a similar viewpoint in defining the transmission capacity with multiple data stream transmission as follows. Similar definition can also be found in \cite{Marios2009, McKayTC2011}.

Let the rate of transmission on each data stream be $R$ bits/sec/Hz. Then the transmission is considered successful on any stream if the 
per-stream SIR  is above a  threshold $\beta$ (function of  $R$). 
Consequently, the  per-stream outage (failure) event  is defined as the event that the SIR on that stream is below a  threshold $\beta$, and the outage probability for the $\ell^{th}$ stream without CSIT is defined to be 
\begin{eqnarray}\label{poutclassicalcsir} 
P_{out}(\ell) &\bydef& P\left(\SIR_{{\cal S}_0, \ell} \le   \beta\right), \ \ell =1,2,\dots, k,
\end{eqnarray}
 and with CSIT to be  
\begin{eqnarray}\label{poutclassicalcsit}
P_{out}^{BF}(\ell) &\bydef& P\left(\SIR_{{\cal S}_0, \ell}^{BF} \le   \beta\right),  \ \ell =1,2,\dots, k.
\end{eqnarray}
Since all the $k$ data streams are independent, we model each transmitter-receiver link as $k$ bit pipes (interfering) with each pipe operating at  $R$ bits/sec/Hz. 
Let  $\epsilon$ be the outage probability constraint of each bit pipe, i.e. $P_{out}(\ell) \le \epsilon$ or $P_{out}^{BF}(\ell) \le \epsilon$ at rate $R$ bits/Hz. Then combining the $k$ data streams, the average throughput for each transmitter-receiver link is $k(1-\epsilon)R$.

 From (\ref{eq:sircsir}), note that without CSIT, $\SIR_{{\cal S}_0, \ell}$ is identically distributed for $\forall \ \ell, \ell =1,2,\dots,k$. Thus, without CSIT, outage probability of any stream can be used for defining the transmission capacity.  With CSIT, however, from (\ref{eq:sircsit}), we can see that the 
$\SIR_{{\cal S}_0, \ell}^{BF}$ is a decreasing function of $\ell, \ \ell =1,2,\dots, k$, since eigenvalues $\gamma_{\ell}$ are 
indexed in the decreasing order, and therefore $P_{out}^{BF}(\ell)$ is not identically distributed for $\ell, \ \ell=1,\dots,k$. Hence with CSIT, to account for the worst case scenario,  we use $P_{out}^{BF}(k)$ for defining the transmission capacity, since it provides 
an upper bound on $P_{out}^{BF}(\ell), \ \ell=1,2\dots,k$.
The formal definitions of the transmission capacity for multiple antennas with and without CSIT are  as follows.

\begin{defn}\label{def:tc}
Without CSIT, since $P_{out}(\ell)$ is identically distributed for $\ell =1,2,\dots, k$, the multiple antenna transmission capacity without CSIT  is defined as
$C_{\epsilon} \bydef k\lambda^{\star}(1-\epsilon)R$, where $\lambda^{\star} = \arg \max_{\lambda} \{P_{out}(1)\le \epsilon  \}$.
\end{defn}
\begin{defn}\label{def:tcbf}
With CSIT, as described before, $P_{out}^{BF}(k)$ provides 
an upper bound on $P_{out}^{BF}(\ell), \ \ell=1,2\dots,k$. Thus accounting for the worst case scenario, 
the transmission capacity with CSIT  and  multiple antennas is defined as 
$C_{\epsilon}^{BF} \bydef k\lambda_{BF}^{\star}(1-\epsilon)R$,
where $\lambda_{BF}^{\star} = \arg \max_{\lambda} \{P_{out}^{BF}(k)\le \epsilon\}$.
\end{defn}

In the next two sections we derive upper and lower bounds on the transmission capacity with and without CSIT to derive the optimal $k$ and $m$ that maximize the transmission capacity.

\section{Without CSIT}
\label{sec:csir} 
In this section we derive upper and lower bounds on the multiple antenna  transmission capacity without CSIT. 
 We consider both the interference cancelation algorithms, canceling the $c(k,m)$ nearest 
interferers, as well as using the CSMIR algorithm, and find the optimal number of streams $k$, and the optimal SRDOF for  interference cancelation $m$, that maximize the transmission capacity without CSIT. 

\subsection{Canceling the Nearest Interferers}
\label{sec:csirclosest}
To calculate the transmission capacity, we first calculate the outage probability defined in (\ref{poutclassicalcsir}).
From (\ref{poutclassicalcsir}),  $P_{out}(\ell)$ is identically distributed $\forall \ \ell$, we drop the index $\ell$ from $P_{out}(\ell)$ and 
write it as $P_{out}$. Without CSIT, from (\ref{eq:sircsirnc}) and  (\ref{poutclassicalcsir}), the outage probability while canceling the $c(k,m)$ nearest interferers at the receiver is 
\begin{eqnarray*}
P_{out} &=& P\left(\frac{d^{-\alpha} s }{I_{sum} } \le \beta\right),
\end{eqnarray*}
where $s$ is  Chi-square distributed with $2(N-m-k+1)$ DOF, and $I_{sum} \bydef \sum_{n=c(k,m)+1}^{\infty}d_n^{-\alpha} \rho_n$, $d_n\le d_m, \ n<m$, and $\rho_n$ is  Chi-square distributed with $2k$ DOF. An upper and lower bound on $P_{out}$ is presented  in the following Theorem to obtain a lower and upper bound on the transmission capacity, respectively.

\begin{thm}
\label{thm:upboundpoutclosest} Without CSIT  and canceling the $c(k,m)= \left\lfloor\frac{m}{k}\right\rfloor$ nearest interferers using partial ZF at the receiver, the outage probability is bounded by
\begin{eqnarray*}
P_{out} &\ge&
\left\{
\begin{array}{ll}
1-\frac{N-m-k+1}{(k-1)d^{\alpha} \beta\left(\pi\lambda\right)^{\frac{\alpha}{2}}}
\left( \left\lfloor\frac{m}{k}\right\rfloor+\frac{5}{8}+\frac{\alpha}{4}\right)^{\frac{\alpha}{2}}, & \text{for} \ k>1, \\
1-\frac{(N-m)}{d^{\alpha} \beta\left(\pi\lambda\right)^{\frac{\alpha}{2}}}
\left( m+\frac{13}{8}+\frac{\alpha}{4}\right)^{\frac{\alpha}{2}} ,  & \text{for} \ k=1,
\end{array}\right. \\
P_{out}&\le& \left\{
\begin{array}{ll} 1-e^{-\beta (\tau +d^{\alpha})\left(\pi\lambda \right)^{\frac{\alpha}{2}}
\left(\frac{\alpha}{2}-1\right)^{-1}\left( \left\lfloor\frac{m}{k}\right\rfloor+1\right)^{1-\frac{\alpha}{2}}
}, \ \ \ \ \ \ \ \ \ \ \ \ \ \ k+m=N, &  \\
\frac{k\beta d^{\alpha} \left(\pi\lambda \right)^{\frac{\alpha}{2}}}{N-m-k}
\left(
\left(\frac{\alpha}{2}-1\right)^{-1}\left( \left\lfloor\frac{m}{k}\right\rfloor+1\right)^{1-\frac{\alpha}{2}}
\right),   \ \ \ \  \ \ \ \ \ \ \ \ \ \text{otherwise.} 
\end{array}
\right. 
\end{eqnarray*}
\end{thm}
\begin{proof}
The outline of the proof is as follows. To derive a lower bound, we consider the interference contribution
from only the nearest non-canceled interferer (the $c(k,m)+1^{st}$
interferer) $I_{c(k,m)+1}\bydef d_{c(k,m)+1}^{-\alpha}\rho_{c(k,m)+1}  $,
since $I_{c(k,m)+1} < I_{sum}$, and 
$P_{out} = \left( \frac{d^{-\alpha} s }{I_{sum}} \le 2^R-1\right) \ge \left( \frac{ d^{-\alpha}s } {I_{c(k,m)+1}} \le 2^R-1\right).$ To derive an upper bound,
we use the Markov inequality with $\frac{I_{sum}}{s}$ as the random
variable.
With $\bbE\{s\} = N-m-k+1$, since $s$ is distributed as Chi-square with $2(N-m-k+1)$ DOF, the detailed proof is derived in Appendix \ref{lboundpoutclosest},  and Appendix \ref{upboundpoutclosest}, respectively. \end{proof}
Using Theorem \ref{thm:upboundpoutclosest}, the optimal $k$ and $m$ that provide the best scaling of transmission
capacity with respect to $N$ is given by the next Corollary.

\begin{cor}\label{cor:upboundpoutclosest} Without CSIT and canceling the $ \left\lfloor\frac{m}{k}\right\rfloor$ nearest interferers, using a single transmit antenna ($k=1$) and a fraction of
the total SRDOF for interference cancelation ($m=\theta N, \ \theta \in (0,1)$)
maximizes the upper and lower bound on the transmission capacity, and provides with the best scaling of the transmission capacity with the number of antennas $N$.
\end{cor}
\begin{proof}
Let $P_{out}=\epsilon$, then from Theorem \ref{thm:upboundpoutclosest},
\begin{eqnarray*}
C_{\epsilon} &\le & \left\{\begin{array}{ll}
\frac{kR(1-\epsilon)^{1-\frac{2}{\alpha}}}{\pi}\left(
\frac{N-m-k+1}{(k-1)d^{\alpha} \beta}\right)^{\frac{2}{\alpha}}
\left(\left\lfloor\frac{m}{k}\right\rfloor +\frac{5}{8}+\frac{\alpha}{4}\right), & \text{for} \ k>1, \\
\frac{R(1-\epsilon)^{1-\frac{2}{\alpha}}}{\pi}\left(\frac{(N-m)}{d^{\alpha} \beta}\right)^{\frac{2}{\alpha}}\left(m +\frac{13}{8}+\frac{\alpha}{4}+1\right) ,  & \text{for} \ k=1,
\end{array}\right. \\
C_{\epsilon}&\ge &
\left\{\begin{array}{ll}

\frac{kR(1-\epsilon)}{\pi}\left(\frac{-\ln(1-\epsilon)}{k \beta d^{\alpha} }\right)^{\frac{2}{\alpha}}\left(
\left(\frac{\alpha}{2}-1\right)^{-1}\left(\left\lfloor\frac{m}{k}\right\rfloor+1\right)^{1-\frac{\alpha}{2}}
 \right)^{-\frac{2}{\alpha}},   \ \ \text{for} \ k+m=N, & \\
\frac{kR(1-\epsilon)}{\pi}\left(\frac{\left(N-k-m\right)\epsilon}{k \beta d^{\alpha} }\right)^{\frac{2}{\alpha}}\left(
\left(\frac{\alpha}{2}-1\right)^{-1}\left(\left\lfloor\frac{m}{k}\right\rfloor+1\right)^{1-\frac{\alpha}{2}}
\right)^{-\frac{2}{\alpha}}, \ \ \ \text{otherwise}. &
\end{array}\right.
\end{eqnarray*}

We evaluate the upper and lower bound for different values of $k$ and $m$ to identify the scaling behavior of the transmission capacity as follows.
\begin{itemize}
\item $k=1, m=c$, where $c$ is a constant that does not depend on $N$:  \ $C_{\epsilon} = \Theta\left(N^{\frac{2}{\alpha}}\right)$. Note that same scaling is obtained for any constant value of $k$ and $m$.
\item  $k=1, m=N-c$, where $c$ is a constant that does not depend on $N$: \  $C_{\epsilon} = \Omega\left(N^{1-\frac{2}{\alpha}}\right)$, and $C_{\epsilon} = {\cal O}\left(N\right)$.
\item $k=1, m=\theta_1N^p, \ p\in [0,1], \ \theta_1 \in (0,1]$:  \ $C_{\epsilon} = \Omega\left(N^{\frac{2}{\alpha} +p\left(1-\frac{2}{\alpha}\right)}\right)$, 
and $C_{\epsilon} = {\cal O}\left(N^{\frac{2}{\alpha} +p}\right)$. 
The upper and lower bound is maximized at $p=1$, and results in  $C_{\epsilon} = \Omega\left(N\right)$, and $C_{\epsilon} = {\cal O}\left(N^{\frac{2}{\alpha} +1}\right)$.

\item $k=\theta_2N^t, m=\theta_1N^p, \ p,t \in [0,1], \ \theta_1, \theta_2 \in (0,1]$: If $p <t$,  $C_{\epsilon} = \Theta\left(N^{t\left(1-\frac{2}{\alpha}\right)+\frac{2}{\alpha}}\right)$, which is maximized at $t=1$ and results in $C_{\epsilon} = \Theta\left(N\right)$, else $C_{\epsilon} = \Omega\left(N^{\frac{2}{\alpha} +p\left(1-\frac{2}{\alpha}\right)}\right)$, and $C_{\epsilon} = {\cal O}\left(N^{p+ \frac{2}{\alpha}(1-t)}\right)$. The upper and lower bound is maximized at $p=1$, and $t=0$,  and results in  
$C_{\epsilon} = \Omega\left(N\right)$, and $C_{\epsilon} = {\cal O}\left(N^{\frac{2}{\alpha} +1}\right)$.
\end{itemize}


Thus, with $k$ being any constant independent of $N$, and $m= \theta N, \ \theta \in (0,1]$ provides the best 
scaling of transmission capacity with $N$. Note that the derived upper bound on the transmission capacity is larger for $k=1$ than for $k>1$, hence the upper bound on the transmission capacity is maximized at $k=1$. Moreover,  the lower bound on the transmission capacity is proportional to $\frac{k^{1-2/\alpha} (N-k)^{2/\alpha}} {((\left\lfloor\frac{\theta N}{k}\right\rfloor+1)^{1-\alpha/2}+\kappa_2)^{\frac{2}{\alpha}}}$, for small constant $\kappa_2$. Thus, the lower bound on the transmission capacity for $k=1$ is greater than for $k>1$. 
Hence, $k=1$ maximizes the derived upper and lower bound on the transmission capacity.\footnote{ For example, in Fig. \ref{fig:lbubvsk}  we plot the upper and lower bound on the transmission capacity for $N=10$ as a function of $k$ for  $m = \min \left\{ \left(1-\frac{2}{\alpha}\right)N, \left\lfloor\frac{N-k}{k}\rfloor\right)   \right\}$, $d=1$ m, $\alpha=3$, $\beta=1$ and $\epsilon=0.1$.}

With $k=1$ and $m= \theta N, \ \theta \in (0,1]$,
\[C_{\epsilon} \ge \frac{NR(1-\epsilon)}{\pi} (1-\theta)^{\frac{2}{\alpha}}
\theta^{1-\frac{2}{\alpha}} \left(\frac{\left(1-\frac{1}{N(1-\theta)}\right)\epsilon}{ \beta d^{\alpha}\left(
\left(\frac{\alpha}{2}-1\right)^{-1}
+ \left(\frac{1}{\theta}\right)^{1-\frac{2}{\alpha}}\sum_{i=1}^{ 1+\left\lceil\frac{\alpha}{2}\right\rceil }\sigma\left(i\right) \right)}\right)^{\frac{2}{\alpha}}
.\]
For large $N$,
to find the optimal value of $\theta_1$ that maximizes the lower bound, we
need to maximize the function $(1-\theta)^{\frac{2}{\alpha}}
\theta^{1-\frac{2}{\alpha}}$. Solving by setting the derivative to zero,
the optimal value of $\theta = 1-\frac{2}{\alpha}$. Note that the lower bound on the transmission capacity is concave in $m$. Thus, to enforce the 
integer constraint on $m$, $m$ should be chosen as $\left\lfloor \left(1-\frac{2}{\alpha}\right)N\right\rfloor$ or  $\left\lceil\left(1-\frac{2}{\alpha}\right)N\right\rceil$ depending on whichever value maximizes the lower bound.
\end{proof}

{\it Discussion:} In this subsection we showed that without transmit beamforming, transmitting a single data stream $k=1$ and using $m= \theta N, \theta\in(0,1]$ SRDOF with partial ZF for canceling the nearest interferers maximizes the derived upper and lower bound on the transmission capacity.  With this optimal choice of $k$ and $m$, the transmission capacity scales linearly with  the number of antennas $N$. Our result is a generalization of \cite{Jindal2008a}, where $k$ was
fixed to $1$ and a lower bound on the transmission capacity was
shown to scale linearly with $N$, for $m=\theta N$.  Thus it is optimal to use a single
transmit antenna ($k=1$) even when there are multiple transmit
antennas available at the transmitter. 
Compared to the sublinear scaling of the transmission capacity with the number of receive antennas when no interference 
cancelation is used \cite{Hunter2008}, we show that by using interference cancelation, the transmission capacity scales linearly with the number of receive antennas. Thus, our result confirms the importance of  interference cancelation in ad-hoc networks, and 
reveals that significance performance gains can be achieved from its application.

The physical interpretation of our result is as follows. Transmitting $k$ data streams simultaneously, 
the transmission capacity is $k$ times the per stream transmission capacity. Increasing $k$, however, decreases the per stream transmission capacity because of : 1) reduced power of the signal of interest, which is distributed as Chi-square with $2(N-k-m+1)$ DOF, 2) increased power of each interferer which is distributed as Chi-square with $2k$ DOF, since $k$ data streams are transmitted simultaneously, and 3) reduction in the number of interferers $c(k,m)\sim \frac{m}{k}$ that can be canceled. The reduction in per stream transmission capacity with increasing $k$ outweighs the linear increase in $k$ because of simultaneous transmission of $k$  data streams, and $k=1$ provides the best scaling of transmission capacity with $N$. At the receiver, using $m$ SRDOF for interference cancelation, the power of the signal of interest $\sim$ Chi-square with $2(N-m-k+1)$ DOF, and the number of interferers that can be canceled is $\sim \frac{m}{k}$. Thus using $k=1, m= \theta N, \theta\in(0,1)$, allows the power of the signal of interest and the number of interferers canceled to grow linearly with $N$, and provides the best scaling of the transmission capacity.

\subsection{Cancelation Using the CMSIR Algorithm}
\label{sec:csirmsir}
From (\ref{eq:sirmsir}) and (\ref{poutclassicalcsir}), with the CMSIR algorithm, the outage probability for the $\ell^{th}$ stream 
$P_{out}^{CMSIR}(\ell) = P(\SIR^{CMSIR}_{ {\cal S}^{\star}_{0\ell}} \le \beta)$ is identically distributed $\forall \ \ell = 1,2,\dots, k$. Hence we drop the index $\ell$ from $P_{out}^{CMSIR}(\ell) $, and represent it 
as $P_{out}^{CMSIR} = P\left(\SIR^{\text{CMSIR}}\le \beta\right)$, where 
\[\SIR^{\text{CMSIR}} =  \max_{ 
{\cal S}, \ {\cal S} \subset \Phi\backslash \{T_0\}, \ |{\cal S}| = c(k,m)}\frac{ d^{-\alpha}  s }{\sum_{n: T_n \in \Phi \backslash \{T_0, \cal S\} }  d_n^{-\alpha}\rho_n}.
\] 
Following Remark \ref{restrictionSIR}, since CMSIR algorithm is better than canceling the nearest interferers, $\SIR^{\text{CMSIR}}$ is greater than or equal to the SIR while canceling the nearest interferers as in Section \ref{sec:csirclosest}.
Therefore the outage probability with the CMSIR algorithm  is upper bounded by the outage probability while canceling the nearest interferers. As a result, the  outage probability with the CMSIR algorithm is upper bounded by the upper bound derived in Theorem \ref{thm:upboundpoutclosest}.  Finding a lower bound on the outage probability with the CMSIR algorithm is comparatively non-trivial. Both bounds are summarized in the next Theorem.

\begin{thm}
\label{thm:mic} Without CSIT  and canceling the $c(k,m)= \left\lfloor\frac{m}{k}\right\rfloor$ interferers using the CMSIR algorithm, the outage probability bounds are given by Theorem \ref{thm:upboundpoutclosest}.
\end{thm}
\begin{proof}
The upper bound follows from Theorem \ref{thm:upboundpoutclosest}, while the lower bound is derived in Appendix \ref{app:mic}.
\end{proof}
Note that the upper and lower bounds derived in Theorem \ref{thm:mic} are identical to that of Theorem \ref{thm:upboundpoutclosest}, hence, we get the transmission capacity upper and lower bounds are identical to that of Corollary \ref{cor:upboundpoutclosest}.

\begin{cor}\label{cor:upboundpoutmsirclosest} Without CSIT and canceling the $c(k,m)= \left\lfloor\frac{m}{k}\right\rfloor$ interferers with the CMSIR algorithm, using a single transmit antenna ($k=1$) and a fraction of
total SRDOF for interference cancelation ($m=\theta N, \ \theta \in (0,1)$)
maximizes the scaling of the transmission capacity with the number of antennas $N$.
\end{cor}

{\it Discussion:} In this subsection we showed 
that transmitting a single data stream $k=1$  together
with  using $m=\theta N$ SRDOF for interference cancelation provides the best scaling of  the transmission capacity with $N$ when the receiver uses the CMSIR algorithm, and the transmission capacity lower bound scales linearly with the number of antennas $N$. 
Note that this is identical to the transmission capacity scaling obtained  by canceling the nearest interferers ( Section \ref{sec:csirclosest}). 
Thus, the optimal transmission capacity scaling  is invariant to the two  cancelation algorithms considered.  
Since the CMSIR algorithm has a much higher CSI requirement, from a scaling perspective canceling the closest interferers is preferred. 
From a non-asymptotic perspective this is confirmed in the simulations, where we see that the transmission capacity while canceling the nearest interferers is not very inferior to that of the transmission capacity of the  CMSIR algorithm, and most of the gain with multiple antennas can be obtained by canceling the nearest interferers, which is fairly easy to implement in practice.


\section{With CSIT}
\label{sec:csit} 
In this section we consider multi-mode beamforming at the transmitter $T_n$ using CSIT of the direct link between each transmitter and its corresponding receiver $\bH_{nn}$.

With CSIT, where we consider that each transmitter uses multi-mode beamforming, from definition \ref{def:tcbf}, the transmission capacity is 
$C_{\epsilon}^{BF} \bydef k\lambda_{BF}^{\star}(1-\epsilon)R$,
where $\lambda_{BF}^{\star} = \arg \max_{\lambda} \{P_{out}^{BF}(k)\le \epsilon\}$. 
In this section we only consider the case when each receiver cancels the $c(k)$ nearest interferers in terms of their distance from $R_0$. 
Deriving exact analytical expression for $P_{out}^{BF}(k)$ requires knowledge of the distribution of $\gamma_k$, which is known 
but not amenable to analysis. 
To obtain upper and lower bounds on the outage probability, as will be seen later, it suffices to know the expected value of the 
maximum eigenvalue of $\bH_{00}\bH_{00}^*$, $\gamma_1$, and the expected value of the reciprocal of $\gamma_1$.
For large $N$, asymptotic results are available about the the maximum eigenvalue of $\bH_{00}\bH_{00}^*$, $\gamma_1$, which are summarized as follows. 
\begin{lemma}\label{lem:eigwishart}Let $\bH$ be a $N\times N$ matrix with i.i.d. ${\cal CN}(0,1)$ entries, and let $\gamma_1$ be the  maximum eigenvalue of the matrix $\bH\bH^*$. Then 
\begin{enumerate}
\item $\frac{\gamma_{1}}{N} \stackrel{a.s.}\rightarrow 4$, 
where $a.s.$ stands for almost sure convergence, 
\item $ \bbE\{\gamma_{1}\} = \Theta(N)$,
\item $\bbE\left\{\frac{1}{\gamma_{1}}\right\} = \Theta\left(\frac{1}{N}\right)$.
\end{enumerate}
\end{lemma}
\begin{proof}1) follows  from Proposition $6.1$ \cite{Edelman}, while 2) and 3) can be derived using an identical proof to Theorem $6.1$ \cite{Edelman} 
by replacing $\log$ function by identity function, and $1/x$ function, respectively. 
\end{proof}
We illustrate the scaling behavior of $\bbE\{\frac{1}{\gamma_1}\}$ with $N$ in Fig. \ref{fig:recipcalmaxeig}, and find that $\bbE\{\frac{1}{\gamma_1}\}$ is sandwiched between $\frac{1}{3.5N}$ and $\frac{1}{4N}$.

\begin{figure}
\centering
\includegraphics[width=4.5in]{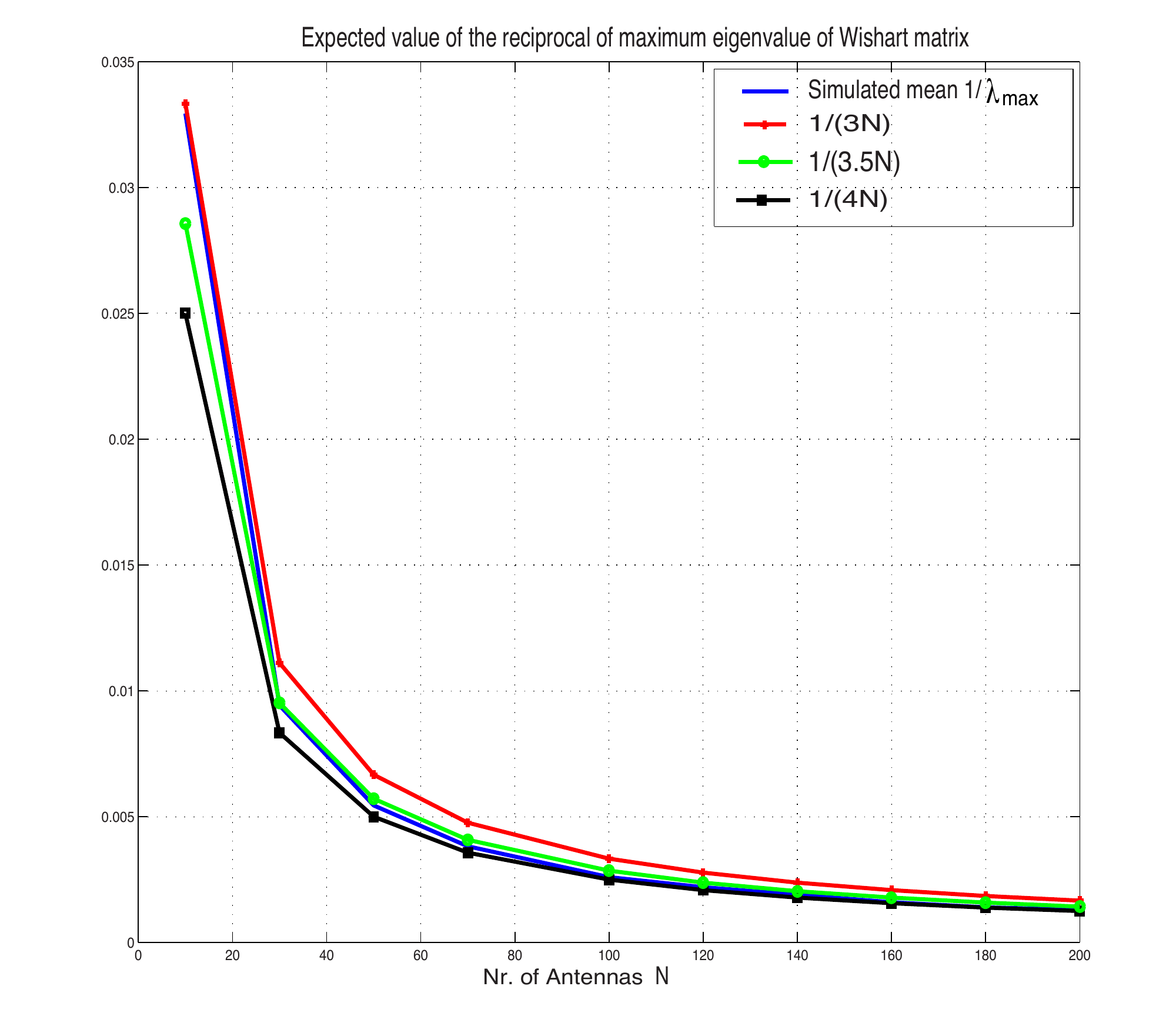}
\caption{Empirical expected value of the reciprocal of the largest eigenvalue  of
$\bH_{00}\bH_{00}^{*}$.}
\label{fig:recipcalmaxeig}
\end{figure}

An upper and
lower bound on the outage probability while canceling the
nearest interferers with CSIT is presented  in the next Theorem.
\begin{thm}
\label{thm:upboundpoutnear} With CSIT, when the transmitter uses multi-mode beamforming
and the receiver cancels the $c(k)=\left\lfloor\frac{N}{k}\right\rfloor-1$ nearest interferers using partial ZF,
the outage probability $P_{out}$ is bounded by
\begin{eqnarray*}
P_{out}^{BF} &\ge&
\left\{
\begin{array}{ll}
1-\frac{\bbE\{\gamma_k\}}{(k-1)d^{\alpha} \beta\left(\pi\lambda\right)^{\frac{\alpha}{2}}}
\left(\left\lfloor\frac{N}{k}\right\rfloor-1+\frac{5}{8}+\frac{\alpha}{4}\right)^{\frac{\alpha}{2}} , & \text{for} \ k>1, \\
1-\frac{\bbE\{\gamma_k\}}{d^{\alpha} \beta\left(\pi\lambda\right)^{\frac{\alpha}{2}}}
\left(\left\lfloor\frac{N}{k}\right\rfloor-1+\frac{13}{8}+\frac{\alpha}{4}\right)^{\frac{\alpha}{2}},  & \text{for} \ k=1,
\end{array}\right. \\
P_{out}^{BF}&\le& \left\{
\begin{array}{c}
(\pi\lambda)^{\frac{\alpha}{2}}k\beta d^{\alpha}
\bbE\{\frac{1}{\gamma_k}\}
\left(
\left(\frac{\alpha}{2}-1\right)^{-1}
\left(
\left\lfloor\frac{N}{k}\right\rfloor\right)^{1-\frac{\alpha}{2}}
\right). \\
\end{array}
\right.
\end{eqnarray*}

\end{thm}

\begin{proof}
The proof is identical to Theorem \ref{thm:upboundpoutclosest}, where $\gamma_k$ takes the role of $s$, and $m=N-k$. 
The lower and upper bound are derived in Appendix \ref{lboundpoutclosest} and Appendix \ref{upboundpoutclosest}, respectively.
\end{proof}
An immediate consequence of Theorem \ref{thm:upboundpoutnear} is that it can be used to find the optimal number of transmitted streams $k$ that maximizes
the transmission capacity.

\begin{cor}
\label{cor:beamnear}
Single stream beamforming ($k=1$) together with canceling
the $c(k)=N-1$ nearest interferers using partial ZF,
maximizes the upper and lower bound on the transmission capacity, and provides the best scaling of the transmission capacity with respect to $N$.

\end{cor}
\begin{proof} For $N\rightarrow \infty$, from Lemma \ref{lem:eigwishart}, $\bbE\{\gamma_1\} \le c_1N$, and $\bbE\left\{\frac{1}{\gamma_{1}}\right\} \ge \frac{c_2}{N}$,  where $c_1$ and $c_2$ are constants. Moreover, since $\gamma_k \le \gamma_1, \ \forall \ k$, 
 $\bbE\{\gamma_k\} \le c_1N$, and $\bbE\left\{\frac{1}{\gamma_{k}}\right\} \ge \frac{c_2}{N}, \ \forall \ k$. 
 Thus, with $P_{out}=\epsilon$,
from Theorem \ref{thm:upboundpoutnear},
\begin{eqnarray*}
C_{\epsilon}^{BF} &\le & \left\{\begin{array}{lc}
\frac{(1-\epsilon)^{1-\frac{2}{\alpha}}kR}{\pi}\left(\frac{c_1N}{(k-1)d^{\alpha}\beta}\right)^{\frac{2}{\alpha}}
\left(\left\lfloor\frac{N}{k}\right\rfloor+\frac{3}{8}+\frac{\alpha}{4}\right), & \text{for} \ k>1, \\
\frac{(1-\epsilon)^{1-\frac{2}{\alpha}}R}{\pi}\left(\frac{c_1N}{d^{\alpha}\beta}\right)^{\frac{2}{\alpha}}\left(N+\frac{5}{8}+\frac{\alpha}{4}\right),  & \text{for} \ k=1,
\end{array}\right. \\
C_{\epsilon}^{BF} &\ge &
\frac{(1-\epsilon)R k^{1-\frac{2}{\alpha}}}{\pi}
\left(\frac{\epsilon}{ \frac{d^{\alpha} \beta c_2}{N}
\left(\left(
\frac{\alpha}{2}-1
\right)\left(\left\lfloor
\frac{N}{k}\right
\rfloor\right)^{1-\frac{\alpha}{2}} \right)}\right)^{\frac{2}{\alpha}}
.
\end{eqnarray*}
Therefore, from the upper and lower bound, using a single data stream $k=1$,
$ C_{\epsilon}^{BF} = \Omega\left(N\right)$, and $C_{\epsilon}^{BF}= {\cal
O}\left(N^{1+\frac{2}{\alpha}}\right)$. For any other constant value of 
$k$ that is independent of $N$ also leads to the same scaling law as $k=1$, however, both the lower and upper bound on the transmission capacity are decreasing functions of $k$. The only option that can
change the scaling behavior of the upper and lower bound is when
$k$ is a function of $N$, $k=\theta N^t, \ \theta \in (0,1], \ t\in [0,1]$. Using
$k=\theta N^t$, for $t=1$, $C_{\epsilon}^{BF} = \Theta(N)$, else, $C_{\epsilon}^{BF} = \Omega(N)$, and $C_{\epsilon}^{BF} = {\cal
O}\left(N^{1+\frac{2}{\alpha}(1-t)}\right)$ which is maximized at $t=0$. Thus, $k=1$ provides the best scaling of the transmission capacity with respect to $N$.
\end{proof}

{\it Discussion:} In this subsection we showed that single stream beamforming ($k=1$) together with canceling the $N-1$ nearest interferers, maximizes the upper and lower bound on the transmission capacity. 
This is in contrast to the findings of Section \ref{sec:csirclosest} where without beamforming the use of $k=1$ and 
$m=\theta N$ SRDOF for canceling the nearest interferers is shown to be optimal. This difference can be explained by noting the fact that with beamforming, for $k=1$, and $N-1$ SRDOF for interference cancelation,  the expected value of the signal power is of order $N$. In comparison, without beamforming, the expected value of the signal power is of order $N-m$, if $m$ SRDOF are used for interference cancelation. Therefore, with no beamforming, using $m=N-1$,  the expected signal power is independent of $N$, and $m=\theta N$ is needed for the expected signal power to grow linearly with $N$. 
Also note that using interference cancelation with single stream beamforming, the transmission capacity scales linearly with $N$ in comparison to the sublinear scaling without interference cancelation \cite{Hunter2008}. Thus, similar to the case of spatial multiplexing without CSIT, interference cancelation is critical even when CSIT is available in an ad-hoc network.

For single stream beamforming, each transmitter requires only the knowledge of the eigenvector corresponding the strongest eigenvalue of the channel between the transmitter and its intended receiver. Thus, the feedback requirements with the optimal strategy  are minimal, and the optimal strategy can be implemented in practice fairly easily.


\section{Simulations}\label{sec:sim} In this section we present numerical simulations to show how the results of this paper (order wise in $N$) apply with finite 
values of number of antennas $N$.
For all the transmission capacity simulations in the paper, the simulated
ad hoc network lies on a two-dimensional disk and contains a number of transmitter-receiver pairs, which
follows the Poisson distribution with the mean equal to $200$. The locations of the nodes are uniformly
distributed on the disk and the disk area is adjusted according to the node density with the typical receiver 
placed at the center of the disk. We set the distance between the typical transmitter and receiver as $d = 1$m or $d=5$m (specified for each plot), the SIR threshold $\beta=1$, and the path-loss exponent $\alpha=3$.

{\bf Without CSIT and canceling the nearest interferers:}
To illustrate the scaling behavior of the transmission capacity with interference cancelation without CSIT, we plot the simulated
transmission capacity for $k=1, m=(1-2/\alpha)N$ (partial ZF), $k=1, m=0$ (MRC), $k=1, m=N-1$, $k=N/2, m=N/2$ (multiple data streams with interference cancelation), and  $k=N/2, m=0$ (multiple data streams with no interference cancelation) in
Fig. \ref{fig:nocsitscalingclosest}. As expected, $k=1, m=(1-2/\alpha)N$ achieves the best performance in terms  of transmission capacity  with increasing $N$, in contrast to all other cases. An important observation from Fig. \ref{fig:nocsitscalingclosest} is that if the number of data streams is scaled linearly with $N$, no interference cancelation leads to better performance compared to performing interference cancelation at the receiver. 
The justification for this result is that when the number of data streams is scaled linearly with $N$, 
employing interference cancelation can remove only  a few interferers, but decreases the signal power significantly.  
For example if $k=N/2$, and $m=N/2$ SRDOF are used for interference cancelation, only one interferer is canceled, and the signal power is Chi-square distributed with $2$ DOF. In comparison with no interference cancelation, the signal power is 
 Chi-square distributed with $2(N/2-1)$ DOF and no interferer is canceled. Since there are a large number of interferers, canceling one  interferer does not significantly increase the SIR, however, the decreased signal power severely affects the performance. 
To illustrate the behavior of the transmission capacity as a function
of the number of data streams $k$, we plot the simulated transmission capacity for $N=10$ by varying $k$ with 
$m = \min \left\{ \left(1-\frac{2}{\alpha}\right)N, \left\lfloor\frac{N-k}{k}\rfloor\right)   \right\}$ in Fig. \ref{fig:csirmsirvsk}. Clearly, the transmission capacity decreases with $k$ as predicted by the derived upper and the lower bounds.

{\bf Without CSIT and canceling using the CMSIR algorithm:}
To illustrate the scaling behavior of the transmission capacity with the CMSIR algorithm as a function
of $N$, we plot the simulated transmission capacity for $k=1, m=(1-\frac{2}{\alpha})N$ in Fig. \ref{fig:csirmsirvsN}. We also illustrate the  performance gain of the CMSIR algorithm in comparison to  canceling the nearest interferers. To illustrate the behavior of the transmission capacity as a function
of the number of data streams $k$, we plot the simulated transmission capacity for $N=10$ by varying $k$ in Fig. \ref{fig:csirmsirvsk}. Clearly, the transmission capacity decreases with $k$ as predicted by the derived upper and the lower bounds.

{\bf With CSIT and canceling the nearest interferers:}
To understand the exact behavior of the transmission capacity with respect to the number of transmitted data streams $k$ when each transmitter uses multi-mode beamforming,  we plot the simulated transmission capacity as a function of $k$ in
Fig. \ref{fig:csitvarkclosest}. From Fig. \ref{fig:csitvarkclosest}
we can see that the transmission capacity is maximized for $k=1$ as
suggested by the derived upper and lower bound, and we conclude that $k=1$ is
optimal for maximizing the transmission capacity when transmitter uses multi-mode beamforming and the receiver cancels the nearest interferers. To illustrate the scaling behavior of the transmission capacity with $N$, we plot
the derived lower and upper bound, and the simulated transmission capacity in Fig.
\ref{fig:csitscalingclosest}, for $k=1$, $d=5$ m, $\alpha =4$ and
$\epsilon = 0.1$ with increasing $N$. 
\section{Conclusions}
\label{sec:conc} In this paper we showed that interference cancelation is critical in ad-hoc networks for maximizing the scaling of the transmission capacity with respect to the number of antennas. A main result of this paper is that  most of the gain from multiple antennas is obtained by using them for  interference cancelation at the receiver. For the case of spatial multiplexing at the transmitter without CSIT, we showed that a single transmit antenna is sufficient at the transmitter, and the transmission capacity increases at least  linearly with the number of receive antennas, by using an appropriate fraction of SRDOF for interference cancelation and leaving the rest of SRDOF for decoding the signal of interest.
With beamforming, we showed that a single stream beamforming is optimal at the transmitter, and
provides with  linear scaling of the transmission capacity with $N$, if all but one of the SRDOF are dedicated for interference cancelation at the receiver. Comparing the beamforming and no-beamforming cases in simulation, we found  that there is not much gain by using beamforming in an ad-hoc network in terms of scaling 
of the transmission capacity with respect to the number of antennas. This is explained by the fact that there is only a constant increase in the signal power (independent of the number of antennas) with beamforming in comparison to the no beamforming case.
\begin{figure}
\centering
\includegraphics[width=5.5in]{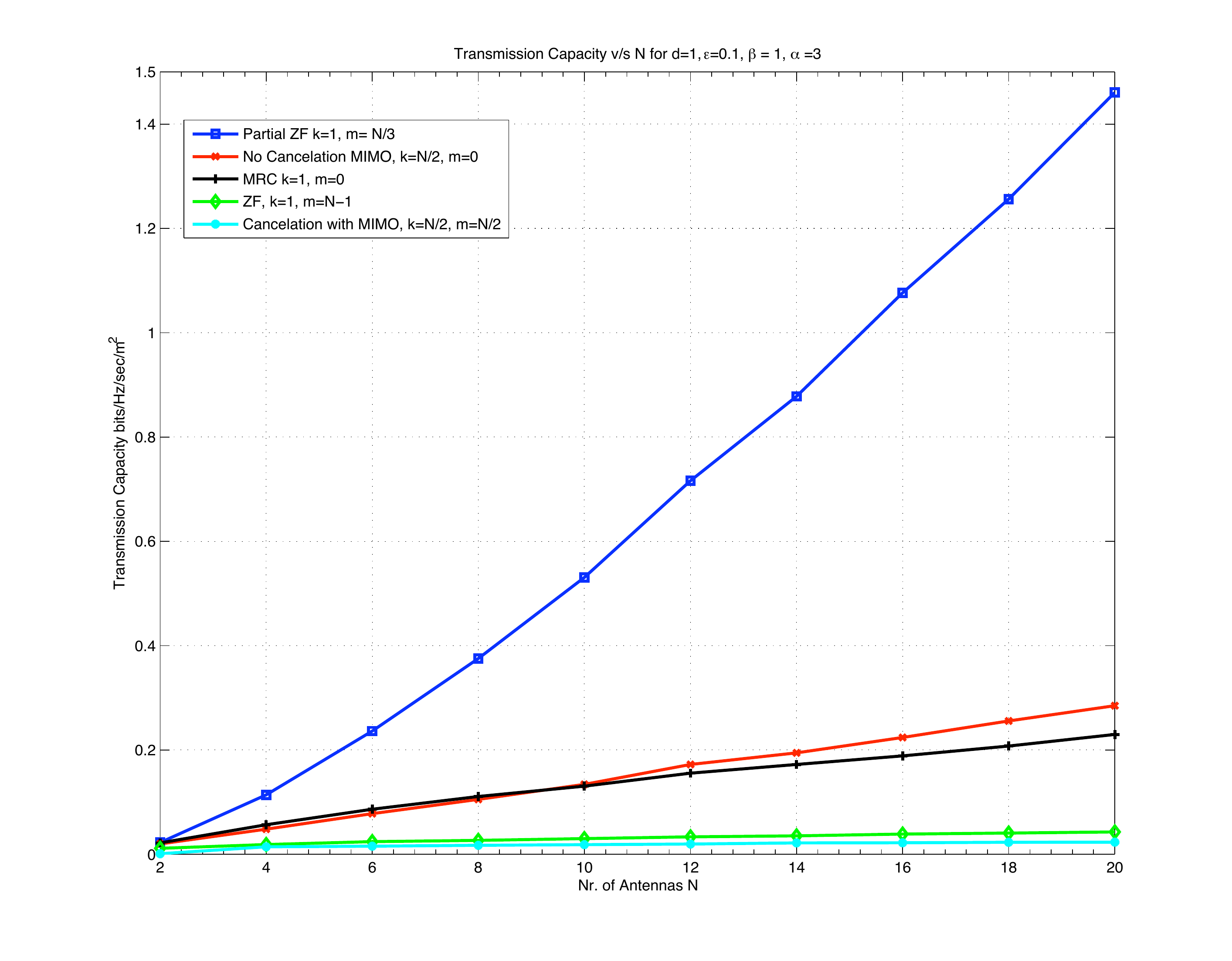}
\caption{Transmission capacity versus $N$ with no beamforming and canceling the nearest interferers for $k=1$, $d=1$ m, $\beta=1$, $\alpha =3$, $\epsilon =0.1$}
\label{fig:nocsitscalingclosest}
\end{figure}

\begin{figure}
\centering
\includegraphics[width=4in]{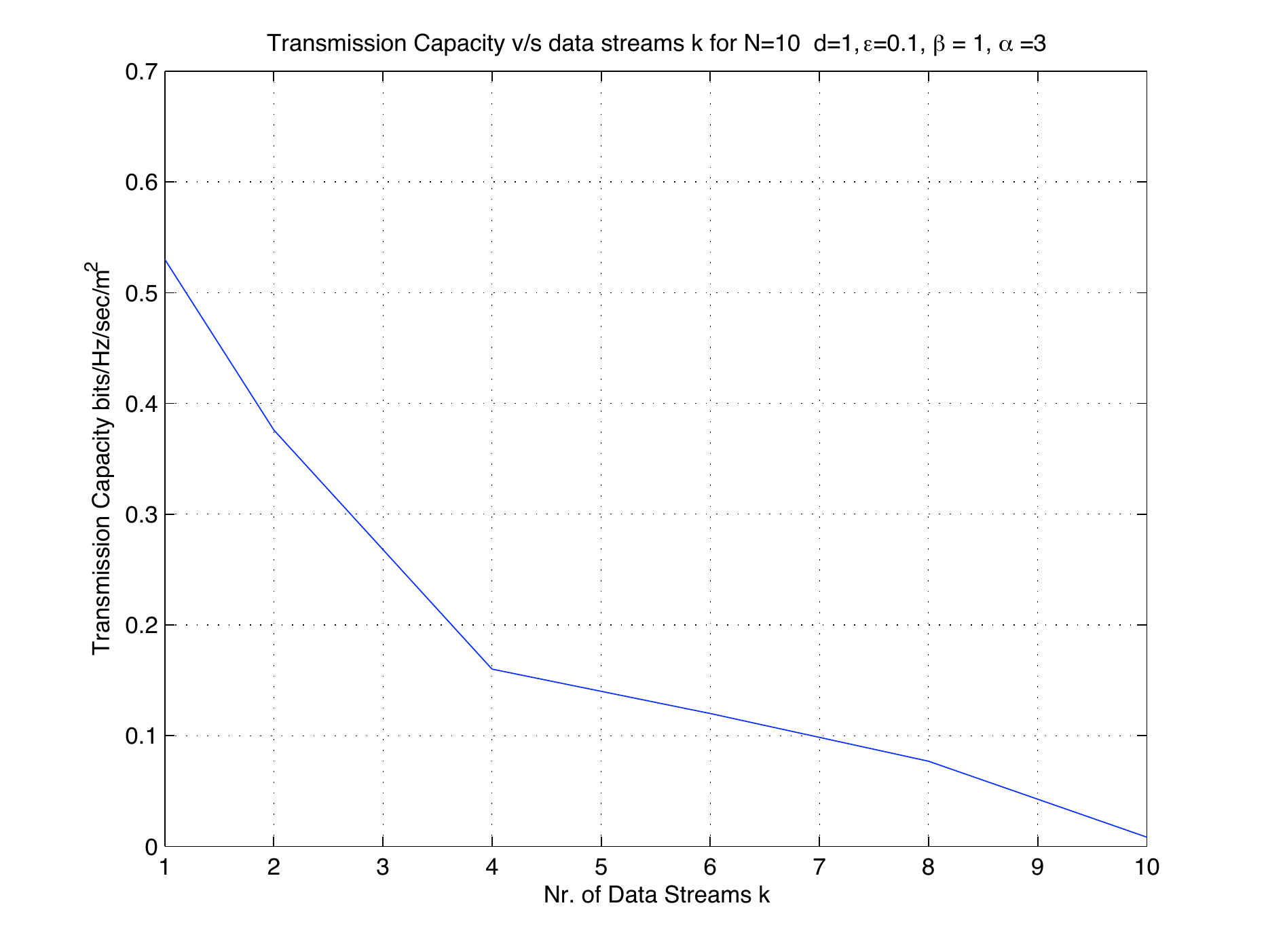}
\caption{Transmission capacity versus the number of transmitted data streams $k$ for $N=10$ with no beamforming and canceling the nearest interferers with   $d=1$ m, $\beta=1$, $\alpha =3$, $\epsilon =0.1$}
\label{fig:csirmsirvsk}
\end{figure}

\begin{figure}
\centering
\includegraphics[width=5.5in]{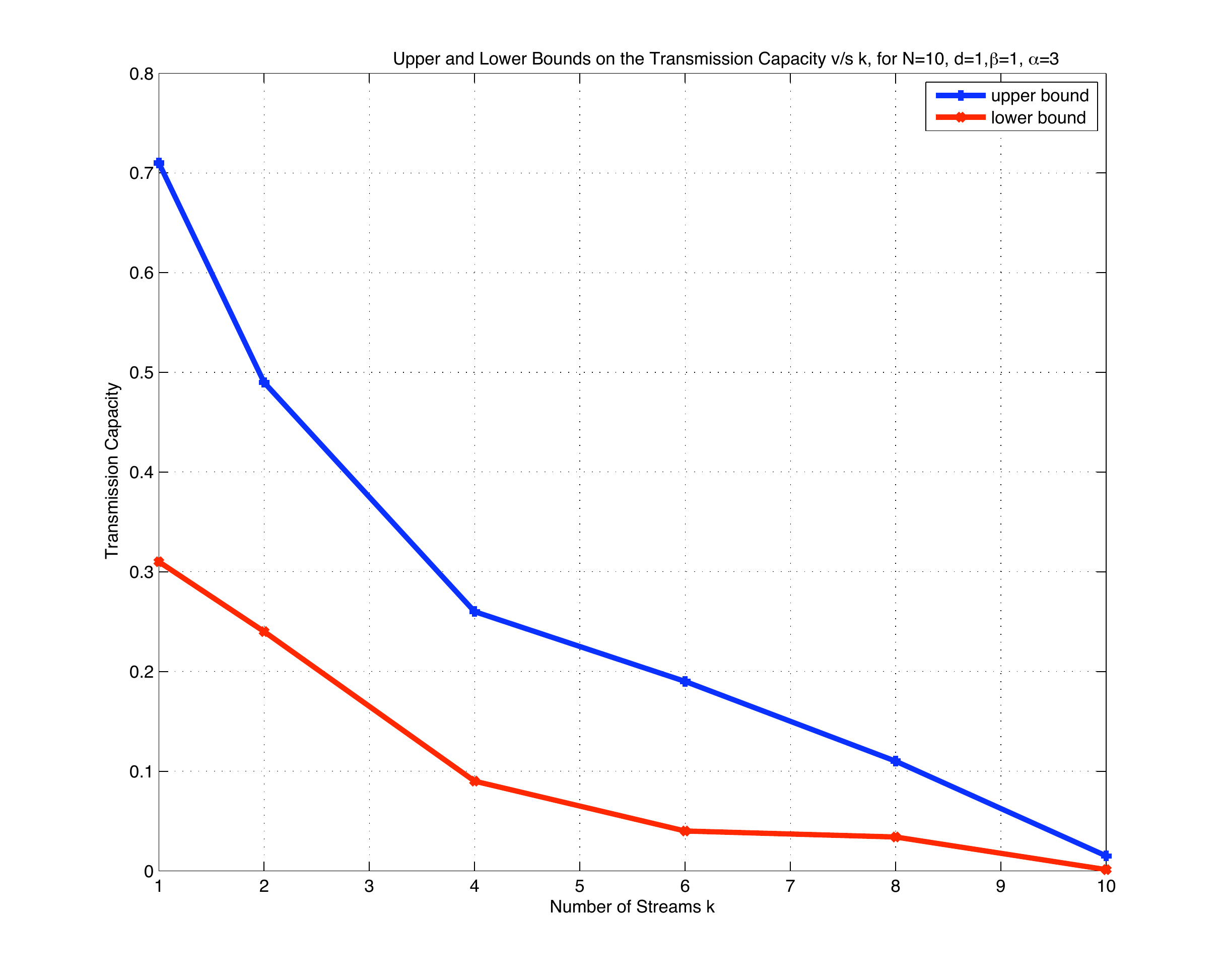}
\caption{Transmission capacity upper and lower bound versus $k$ with no beamforming and canceling the nearest interferers for $N=10$, $d=1$ m, $\beta=1$, $\alpha =3$, $\epsilon =0.1$}
\label{fig:lbubvsk}
\end{figure}

\begin{figure}
\centering
\includegraphics[width=5.5in]{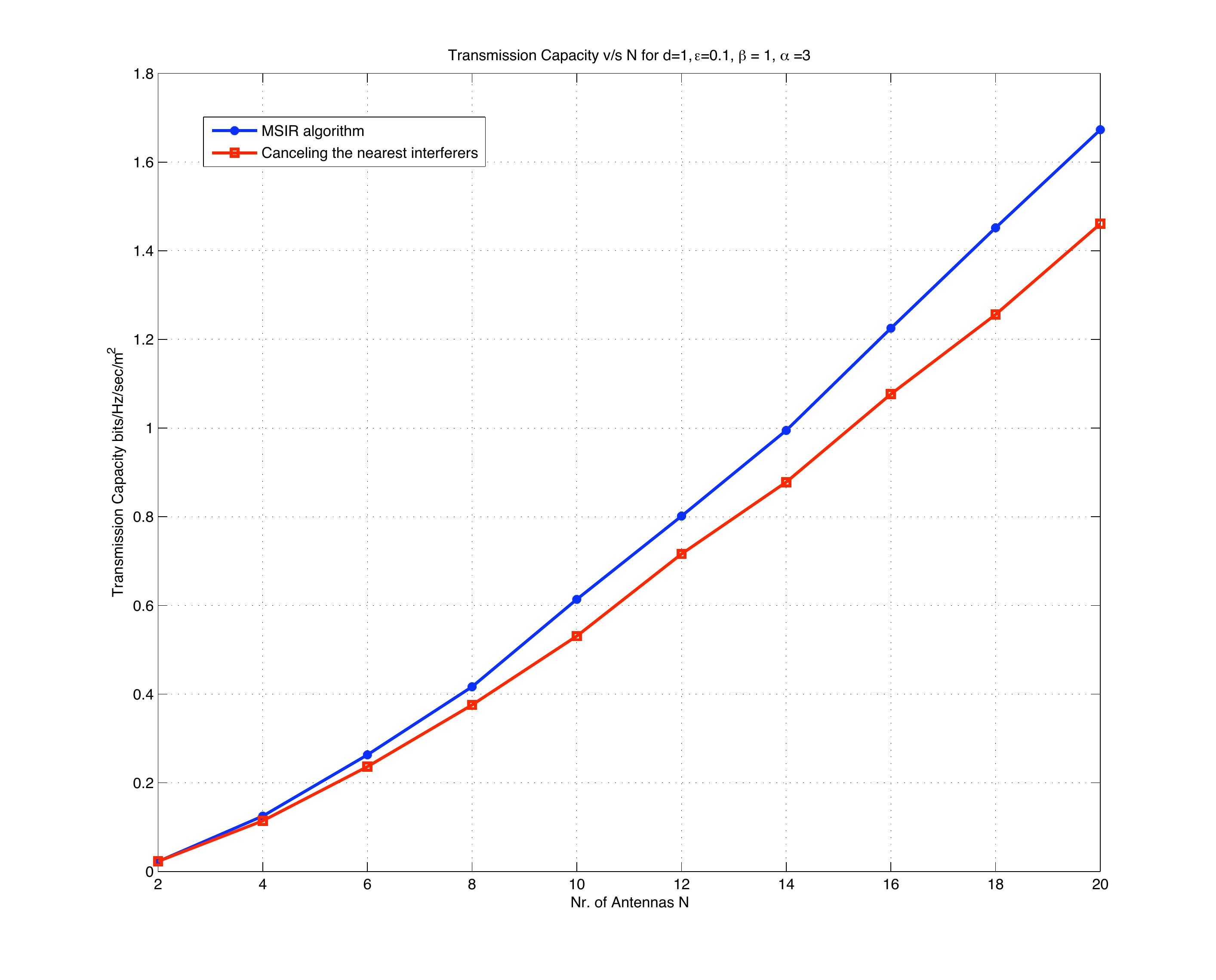}
\caption{Transmission capacity versus $N$ with no beamforming and using the CMSIR algorithm for interference cancelation  for $k=1$, $d=1$ m, $\beta =1$, $\alpha =3$, $\epsilon =0.1$}
\label{fig:csirmsirvsN}
\end{figure}

\begin{figure}
\centering
\includegraphics[width=4in]{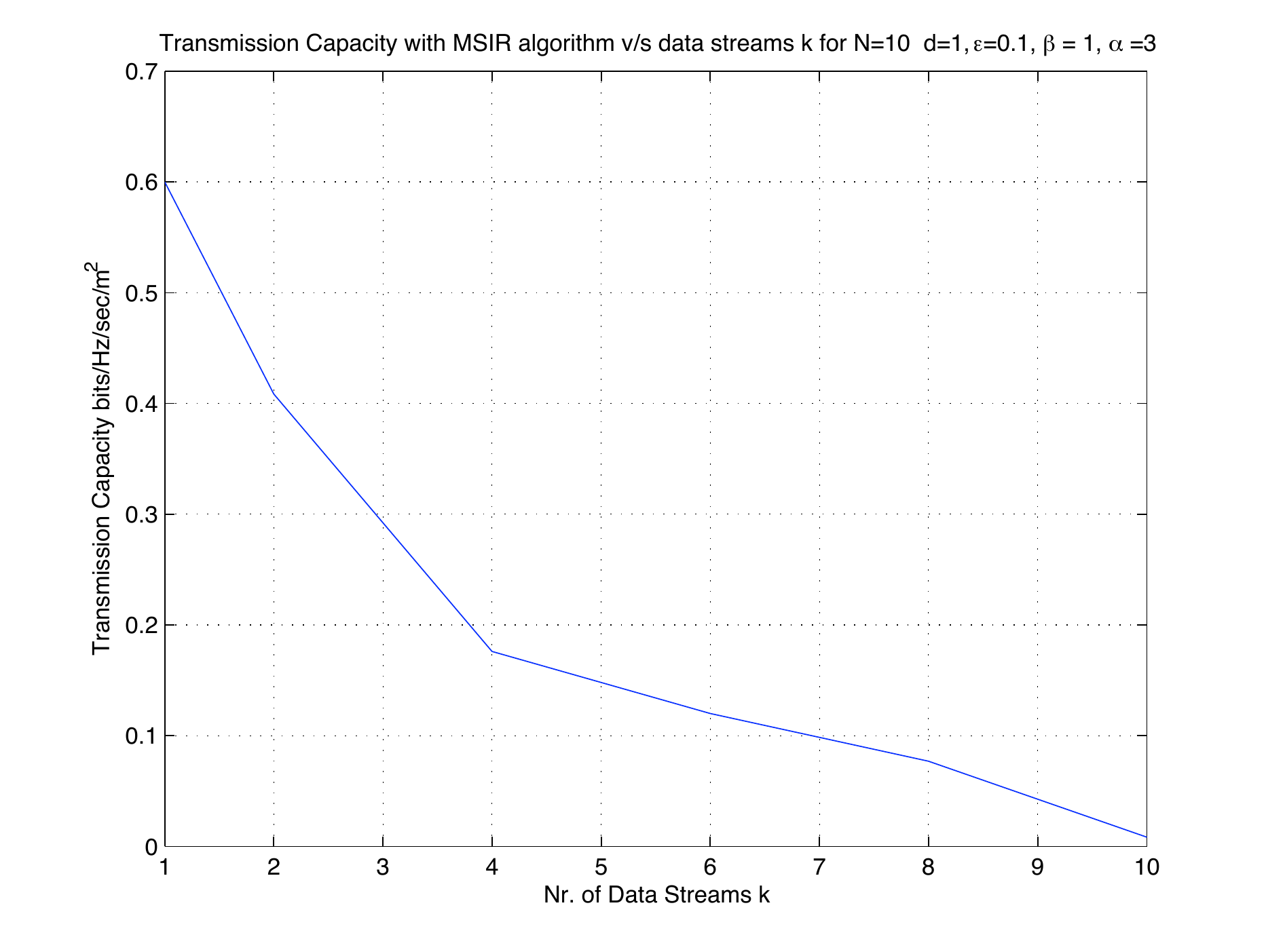}
\caption{Transmission capacity versus the number of transmitted data streams $k$ for $N=10$ with no beamforming and using the CMSIR algorithm for interference cancelation  and  $k=1$, $d=1$ m, $\beta =1$, $\alpha =3$, $\epsilon =0.1$}
\label{fig:csirmsirvsk}
\end{figure}

\begin{figure}
\centering
\includegraphics[width=4in]{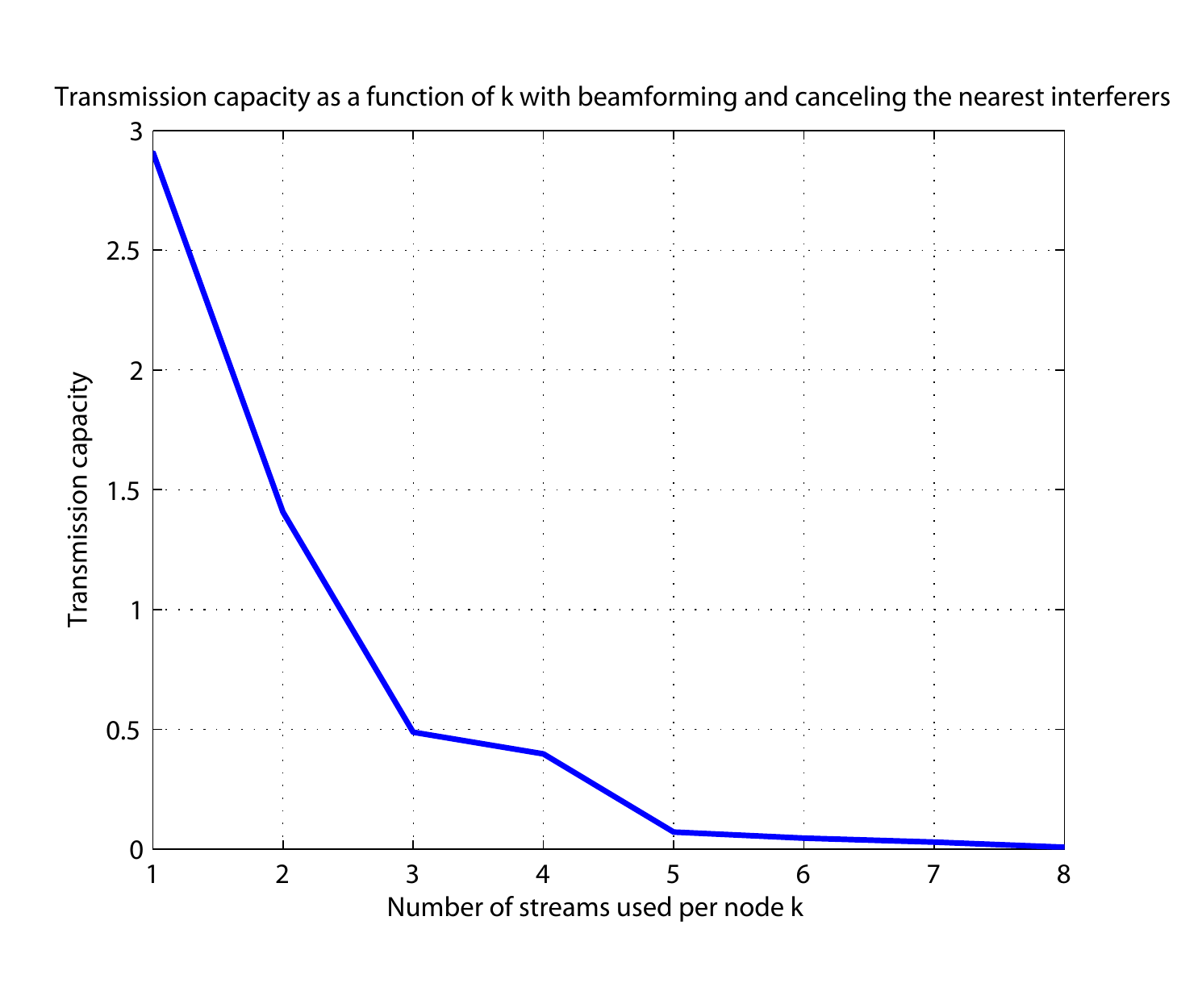}
\caption{Transmission capacity versus $k$ with beamforming and canceling the nearest interferers for $k=1$, $d=1$ m, $\beta =1$, $\alpha =4$, $\epsilon =0.1$, $N=8$.}
\label{fig:csitvarkclosest}
\end{figure}

\begin{figure}
\centering
\includegraphics[width=4in]{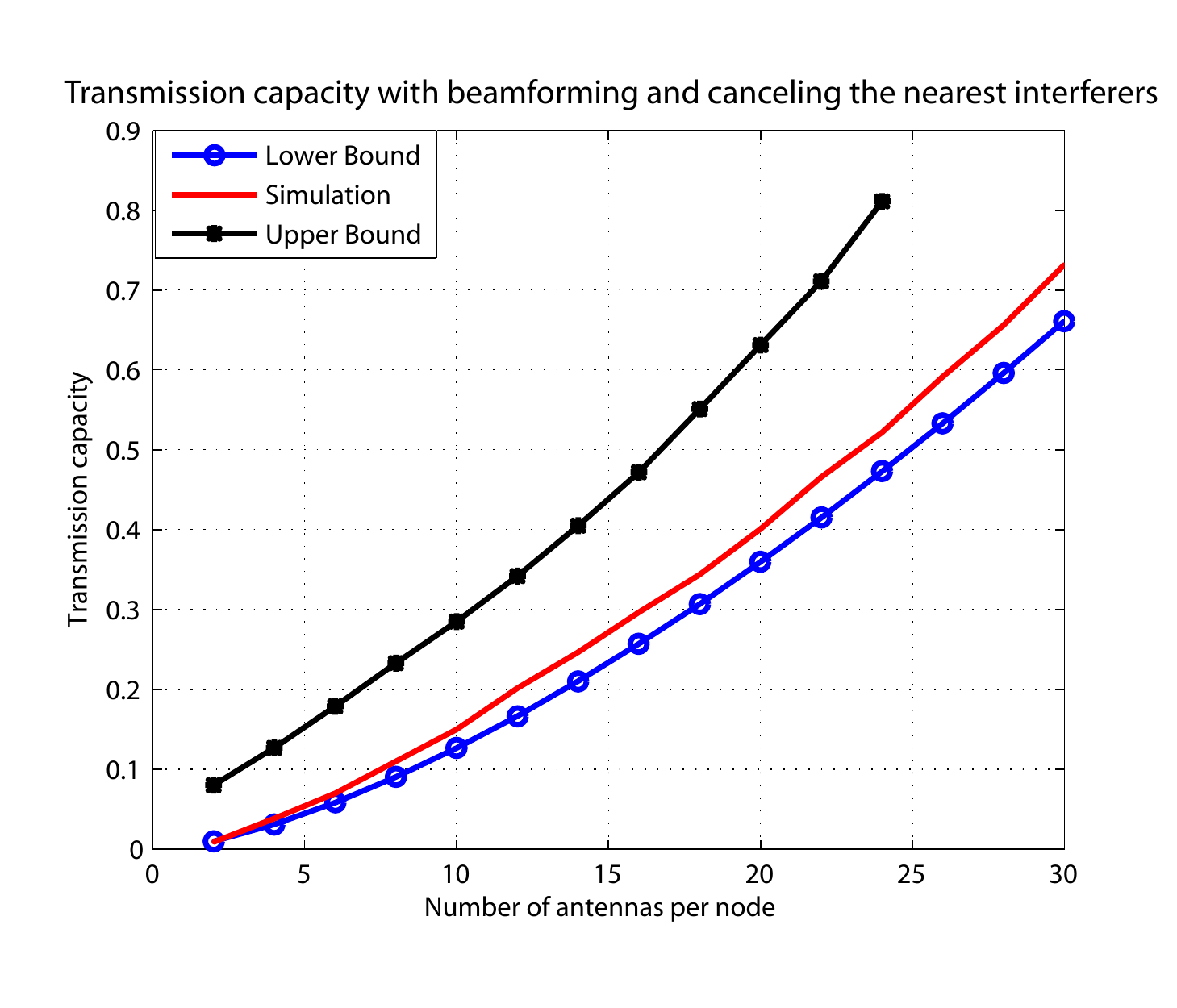}
\caption{Transmission capacity versus $N$ with beamforming and canceling the nearest interferers for $k=1$, $d=5$ m, $\beta =1$, $\alpha =4$, $\epsilon =0.1$}
\label{fig:csitscalingclosest}
\end{figure}

\section{Acknowledgements}
We would like to thank Dr. Marios Kountouris and Dr. Nihar Jindal for useful comments and discussions.
\appendices
\section{Lower bound on the outage probability while canceling the nearest
interferers}
\label{lboundpoutclosest}
Let $s$ be the power of the signal of interest.
For Section \ref{sec:csit}, $s=\gamma_k$.
To lower bound the outage probability we consider the
interference contribution from only the nearest non-canceled interferer, that
is the $c(k,m)+1^{st}$ interferer, as follows.
Recall that
\begin{eqnarray}\nonumber
1-P_{out} &=& P\left(\frac{d^{-\alpha}  s}{I_{sum}}>\beta\right), \\\nonumber
&\le& P\left(\frac{d^{-\alpha}  s}{I_{c(k,m)+1}}>\beta\right), \ \text{since} \ I_{c(k,m)+1}
\le I_{sum}, \\\nonumber
&\le & \frac{\bbE\left\{\frac{s}{I_{c(k,m)+1}}\right\}}{d^{\alpha} \beta}, \ \ \  \ \text{ from the Markov inequality},\\ \label{eq:app1intm1}
&=& \frac{\bbE\{s\}}{d^{\alpha}\beta}\bbE\left\{\frac{1}{I_{c(k,m)+1}}\right\},  \ \text{since} \ s \  \text{and} \  I_{c(k,m)+1}  \ \text{are independent}.
\end{eqnarray}
Recall that $I_{c(k,m)+1} = d_{c(k,m)+1}^{-\alpha}  \rho_{c(k,m)+1}$, where
$\rho_{c(k,m)+1}$ and $d_{c(k,m)+1}$ are independent of each other. Moreover, $\rho_{c(k,m)+1}$
is Chi-square distributed
with $2k$ DOF, and  $\pi \lambda d_{c(k,m)+1}^2$ is Chi-square distributed
with $2(c(k,m)+1)$ DOF \cite{Jindal2008a}. 
Therefore from (\ref{eq:app1intm1})
\begin{eqnarray*}
1-P_{out}
&\le& \frac{\bbE\{s\}}{d^{\alpha}\beta}\bbE\left\{d^{\alpha}_{c(k,m)+1}\right\}\bbE\left\{\frac{1}{\rho_{c(k,m)+1}}\right\}.
\end{eqnarray*}
Since $\rho_{c(k,m)+1}$
is Chi-square distributed
with $2k$ DOF, $\bbE\left\{\frac{1}{\rho_{c(k,m)+1}}\right\} = \frac{1}{k-1}$ for $k>1$, and $\bbE\left\{\frac{1}{\rho_{c(k,m)+1}}\right\} = \infty$ for $k=1$. Thus, this lower bounding technique is useful only for $k>1$. For $k=1$ a lower bound is obtained separately.

{\bf Case 1\ ($k>1$):}
Using change of variables $\bbE\left\{d^{\alpha}_{c(k,m)+1}\right\} = \frac{1}{\left(\pi\lambda\right)^{\frac{\alpha}{2}}} \int_{0}^{\infty} \frac{x^{c(k,m)+\frac{\alpha}{2}}e^{-x}}{c(k,m)!}dx = \frac{
\Gamma\left(c(k,m)+1+\frac{\alpha}{2}\right)}{\Gamma\left(c(k,m)+1\right)}$, and $\bbE\left\{\frac{1}{\rho_{c(k,m)+1}}\right\} = \left(\frac{1}{k-1}\right)$.
Thus,
\begin{eqnarray*}
1-P_{out} 
&\le& \frac{\bbE\left\{ s\right\}}{d^{\alpha}\beta}
\left( \frac{1}{\left(\pi\lambda\right)^{\frac{\alpha}{2}}}\frac{
\Gamma\left(c(k,m)+1+\frac{\alpha}{2}\right)}{\Gamma\left(c(k,m)+1\right)}\right)
\left(\frac{1}{k-1}\right).
\end{eqnarray*}

{\bf Case 2\ ($k=1$):} To obtain a lower bound on the outage probability for $k=1$, we consider the interference contribution of the two nearest non-canceled interferers $I_{c(k,m)+1}$ and $I_{c(k,m)+2}$. Hence
\begin{eqnarray*}
1-P_{out} &=& P\left( \frac{ d^{-\alpha}  s}{I_{sum}}>\beta\right), \\
&\le& P\left(\frac{d^{-\alpha}  s}{I_{c(k,m)+1} + I_{c(k,m)+2}}>\beta\right), \ \text{since} \ I_{c(k,m)+1}+ I_{c(k,m)+2}
\le I_{sum}.
\end{eqnarray*}

By definition $d_{c(k,m)+1} \le d_{c(k,m)+2}$, and hence
$I_{c(k,m)+1} + I_{c(k,m)+2} =d_{c(k,m)+1}^{-\alpha} \rho_{c(k,m)+1} + d_{c(k,m)+2}^{-\alpha} \rho_{c(k,m)+2} 
\ge d_{c(k,m)+2}^{-\alpha} (\rho_{c(k,m)+1}  + \rho_{c(k,m)+2} ) $. 
Therefore
\begin{eqnarray*}
1-P_{out} 
&\le & P\left(\frac{d^{-\alpha}s}{ d_{c(k,m)+2}^{-\alpha} (\rho_{c(k,m)+1}  + \rho_{c(k,m)+2} )}>\beta\right). \\
\end{eqnarray*}
Note that  $\rho_{c(k,m)+1}$ and $\rho_{c(k,m)+2}$
are i.i.d. with Chi-square distribution with $2$ DOF for
$k=1$,  therefore $\rho_{c(k,m)+1} + \rho_{c(k,m)+2}$ is Chi-square distributed with $4$
DOF.  Hence, similar to the case of $k>1$,  
\begin{eqnarray*}
1-P_{out} &\le& \frac{\bbE\left\{ s\right\}}{d^{\alpha}\beta}
\left( \frac{1}{\left(\pi\lambda\right)^{\frac{\alpha}{2}}}\int_{0}^{\infty} \frac{x^{c(k,m)+1 + \frac{\alpha}{2}}e^{-x}}{(c(k,m)+1)!}dx\right) \
\int_{0}^{\infty}e^{-\rho}d\rho, \\
&=& \frac{\bbE\left\{ s\right\}}{d^{\alpha}\beta}
\left(\frac{1}{\left(\pi\lambda\right)^{\frac{\alpha}{2}}}\frac{
\Gamma\left(c(k,m)+2+\frac{\alpha}{2}\right)}{\Gamma\left(c(k,m)+2\right)}\right).
\end{eqnarray*}

From  \cite{Laforgia},
\begin{eqnarray*}
\frac{
\Gamma\left(c(k,m)+\frac{\alpha}{2}\right)}{\Gamma\left(c(k,m)+1\right)} \le
\left(c(k,m)+\frac{1}{8}+\frac{\alpha}{4}\right)^{\frac{\alpha}{2}-1},\end{eqnarray*}
for $c(k,m), \frac{\alpha}{2} >0$.
Therefore, 
\begin{eqnarray*}
1-P_{out} &\le&
\left\{
\begin{array}{cc}
\frac{\bbE\left\{ s\right\}}{(k-1) d^{\alpha}\beta\left(\pi\lambda\right)^{\frac{\alpha}{2}}}
 \left(c(k,m)+\frac{5}{8}+\frac{\alpha}{4}\right)^{\frac{\alpha}{2}}, & k>1, \\
\frac{\bbE\left\{ s\right\}}{d^{\alpha}\beta\left(\pi\lambda\right)^{\frac{\alpha}{2}}}
\left(c(k,m)+\frac{13}{8}+\frac{\alpha}{4}\right)^{\frac{\alpha}{2}},  & k=1.
\end{array}\right.
\end{eqnarray*}

\section{Upper bound on outage probability while canceling the nearest interferers}
\label{upboundpoutclosest}
Let $s$ be the power of the signal of interest. For Sections \ref{sec:csirclosest}, and \ref{sec:csirmsir}, $s$ is Chi-square distributed with $2(N-k-m+1)$ DOF, while 
for Section \ref{sec:csit}, $s=\gamma_k$. 
We obtain upper bounds for the two cases: $N>k+m$, and $N=k+m$, separately  as follows.

{\bf Case 1:  $N>k+m$.}
By definition
\begin{eqnarray}\nonumber
P_{out} &=& P\left(\frac{I_{sum}}{s} \ge \frac{1}{d^{\alpha} \beta}\right),\\\nonumber
&\le &  d^{\alpha} \beta\ \bbE\left\{\frac{I_{sum}}{s}\right\},  \  \text{ from the Markov inequality},\\ \label{eq:app2pout}
 &= &d^{\alpha} \beta\ \bbE\left\{I_{sum}\right\}\bbE\left\{\frac{1}{s}\right\}, \ \text{since} \ s \ \text{and} \  I_{sum}\  \text{are independent}.
\end{eqnarray}

Recall that $I_{sum} = \sum_{i=c(k,m)+1}^{\infty}d_i^{-\alpha} \rho_i$, where $\rho_i$ is a Chi-square random
variable with $2k$ DOF $\forall \ i$. Hence  $\bbE\{I_{sum} \} = \sum_{i=c(k,m)+1}^{\infty}\bbE \left\{ d_i^{-\alpha}\right\}\bbE\{\rho_i\} = k\sum_{i=c(k,m)+1}^{\infty}\bbE \left\{ d_i^{-\alpha}\right\}$. From \cite{Jindal2008a}, 
$\pi \lambda d_{i}^2$ is Chi-square distributed
with $2i$ DOF. Hence $\bbE\{d_{i}^{-\alpha}\}  = (\pi \lambda)^{\frac{\alpha}{2}}\frac{\Gamma(i-\frac{\alpha}{2})}{\Gamma(i)}$, which is finite only if $i> \frac{\alpha}{2}$. This is the restriction because of which we needed the condition that the number of interferers canceled $c(k,m)$ or $c(k)$ is greater than $\frac{\alpha}{2}-1$.  Therefore, as derived in  \cite{Jindal2008a},
\begin{eqnarray}
\label{eq:app2int1}
\bbE \left\{\sum_{i=c(k,m)+\left \lceil \frac{\alpha}{2} \right\rceil+1}^{\infty}d_i^{-\alpha}\right\} &\le&
(\pi\lambda)^{\frac{\alpha}{2}}\left(\frac{\alpha}{2}-1\right)^{-1}\left(c(k,m)\right)^{1-\frac{\alpha}{2}}.
\end{eqnarray}

Hence, 
\begin{equation}
\label{intupboundclosecancl}
\bbE \{I_{sum}\} \le (\pi\lambda)^{\frac{\alpha}{2}}k\left(\frac{\alpha}{2}-1\right)^{-1}\left(c(k,m)+1\right)^{1-\frac{\alpha}{2}},
\end{equation}
and from (\ref{eq:app2pout}),
\begin{eqnarray*}
P\left(I_{sum} \ge \frac{d^{-\alpha}  s}{\beta}\right)&\le& (\pi\lambda)^{\frac{\alpha}{2}}k\beta d^{\alpha} 
\left(
\left(\frac{\alpha}{2}-1\right)^{-1}
\left(
c(k,m)+1\right)^{1-\frac{\alpha}{2}}
 \right) \bbE\left\{\frac{1}{s}\right\}.
\end{eqnarray*}
For Sections \ref{sec:csirclosest}, and \ref{sec:csirmsir}, $s$ is Chi-square distributed with $2(N-k-m+1)$ DOF, and 
hence $\bbE\left\{\frac{1}{s}\right\} = \frac{1}{N-k-m}$, while for Section \ref{sec:csit} signal power $s =\gamma_k$. Clearly, for 
$N=k+m$, this upper bound in meaningless for  Sections \ref{sec:csirclosest}, and \ref{sec:csirmsir}, and we obtain an upper bound for the case $N=k+m$ separately as follows.

{\bf Case 2:  $N=k+m$.} Recall that the signal power $s$ is Chi-square distributed with $2(N-k-m+1)$ DOF. 
Thus, with $N=k+m$, $s$ is exponentially distributed with parameter $1$. Note that this extra step is not required for Section \ref{sec:csit}, where signal power $s =\gamma_k$. Therefore  
\begin{eqnarray}\nonumber
P_{out} &=& P\left(s \le d^{\alpha} \beta I_{sum} \right),\\\nonumber
&=& \bbE\left\{ 1- e^{-\left(d^{\alpha} \beta I_{sum}\right)}\right\},  \  \text{ since $s$ is exponentially distributed},\\ \nonumber
 &\le & 1- e^{-\left(d^{\alpha} \beta \bbE\left\{ I_{sum}\right\}\right)}, \ \text{using} \  \text{Jensen's Inequality, since $e^{(-x)}$ is  convex}.
\end{eqnarray}
Hence we get the following upper bound on the outage probability
 \begin{eqnarray*}
P_{out}&\le&
\left\{
\begin{array}{ll}
1- e^{-\left(d^{\alpha} \beta \bbE\left\{ I_{sum}\right\}\right)}
, \ \ \ \text{for} \ \ k+m=N, &  \\
\frac{d^{\alpha} \beta\bbE\left\{  I_{sum}\right\}}{N-m-k}
,  \ \ \ \ \  \ \ \ \ \ \ \ \ \ \ \text{otherwise},&
\end{array}
\right.
\end{eqnarray*}
where $\bbE\left\{ I_{sum}\right\}$ is given by (\ref{intupboundclosecancl}).

\section{Lower bound on the outage probability with the CMSIR algorithm}
\label{app:mic}
With the CMSIR algorithm, the SIR is 
\[\SIR^{\text{CMSIR}} \bydef \max_{ {\cal S} \subset \Phi\backslash \{T_0\}, \  |{\cal S}| = c(k,m)} \frac{ d^{-\alpha} s  }{\sum_{n:T_n\in \Phi \backslash \{T_0, \cal S\}}  d_n^{-\alpha} \rho_n},\] where $s$ is Chi-square distributed with $2(N-m-k+1)$ DOF, and $\rho_n$ is Chi-square distributed with $2k$ DOF $\forall n \in \Phi \backslash \{T_0, \cal S\}$. The resulting 
outage probability is 
\begin{eqnarray*}P^{\text{CMSIR}}_{out} &=& P\left( \SIR^{\text{CMSIR}}<\beta\right).
\end{eqnarray*}

 Let the  optimal set of interferers to be canceled be ${\cal S}^{\star}$.
Note that the CMSIR algorithm cancels only $c(k,m)$ interferers, therefore at least one out of the $c(k,m)+1$ nearest interferers is not canceled. Assume that the $r^{th}, \ r=1,2,\ldots, c(k,m)+1$ nearest interferer is not canceled, i.e.  
$T_r \in \Phi \backslash \{T_0, \cal S^\star\}$. 
Then the interference $\sum_{n: \ T_n \in \Phi \backslash \{T_0, \cal S^{\star}\}}  d_n^{-\alpha} \rho_n \ge I_r$. 
Since $d_n^{-\alpha}$ is a decreasing function in $d_n$, 
$d_r^{-\alpha} \ge d_{c(k,m)+1}^{-\alpha}, \ r=1,2,\ldots, c(k,m)+1$, and $I_r= d_r^{-\alpha}\rho_r  \ge 
d_{c(k,m)+1}^{-\alpha}\rho_r, \ r=1,2,\ldots, c(k,m)+1$. Moreover, since all the $\rho_n$'s, $n\in \{\Phi \backslash {\cal S}^{\star}\}$'s are identically distributed as Chi-square with $2$ DOF, and $\rho_n$ and $d_n^{-\alpha}$ are independent, 
 $d_{c(k,m)+1}^{-\alpha}\rho_r $ and $I_{c(k,m)+1} = d_{c(k,m)+1}^{-\alpha}\rho_{c(k,m)+1}$ have the same distribution, where $I_{c(k,m)+1}$ is the interference received from the $c(k,m)+1^{th}$ nearest interferer.
 Hence 
\begin{eqnarray*}
P_{out}^{\text{CMSIR}} &\ge & 
P\left( \frac{ d^{-\alpha}  s}
{I_r}
 < \beta
\right),
\\ 
&\ge & 
P\left( 
\frac{ 
d^{-\alpha}  s
}
{
d_{c(k,m)+1}^{-\alpha}\rho_r 
}
 < \beta
\right),\\
&=&P\left( 
\frac{ d^{-\alpha}  s}
{I_{c(k,m)+1}}< \beta\right),\end{eqnarray*}
since similar to $d_{c(k,m)+1}^{-\alpha} \rho_r $, $I_{c(k,m)+1}$is also independent of $s$. Hence we have shown that the outage probability with the CMSIR algorithm is lower bounded by the outage probability while considering only the interference contribution from the $c(k,m)+1^{th}$ nearest interferer. Recall that we have already derived a lower bound on the the outage probability while  considering only the interference contribution from the $c(k,m)+1^{th}$ interferer in Appendix \ref{lboundpoutclosest}, from which we get the required lower bound for the CMSIR algorithm.

\bibliographystyle{../../IEEEtran}
\bibliography{../../IEEEabrv,../../Research}

\end{document}